\documentclass[sigconf,natbib=true]{acmart}

\usepackage{textcomp}
\usepackage{stfloats}
\usepackage{url}
\usepackage{verbatim}
\usepackage{amsthm}
\usepackage{subfigure}
\usepackage{makecell}
\usepackage{multirow}
\usepackage{enumitem}
\usepackage{mathtools}
\usepackage{subcaption}

\usepackage{color, colortbl}

\definecolor{Gray}{rgb}{0.9, 0.9, 0.9}

\usepackage{bm}
\newcommand{\method}{IIMRec}
\newcommand{\R}{\mathbb{R}}

\newcommand{\cB}{\mathcal{B}}
\newcommand{\cN}{\mathcal{N}}
\newcommand{\cG}{\mathcal{G}}
\newcommand{\cU}{\mathcal{U}}
\newcommand{\cI}{\mathcal{I}}
\newcommand{\cL}{\mathcal{L}}
\newcommand{\cD}{\mathcal{D}}

\newcommand{\cM}{\mathcal{M}}
\newcommand{\cE}{\mathcal{E}}
\newcommand{\cH}{\mathcal{H}}
\newcommand{\bW}{\bm{W}}
\newcommand{\bE}{\bm{E}}

\newcommand{\bH}{\bm{H}}
\newcommand{\be}{\bm{e}}
\newcommand{\bx}{\bm{x}}
\newcommand{\bz}{\bm{z}}
\newcommand{\bh}{\bm{h}}

\newtheorem{theorem}{Theorem}[section]

\newtheorem{corollary}[theorem]{Corollary}

\newtheorem{assumption}[theorem]{Assumption}

\usepackage[capitalize,noabbrev]{cleveref}

\usepackage{algorithmic}
\usepackage[ruled,linesnumbered]{algorithm2e}

\usepackage[inkscapelatex=false]{svg}

\copyrightyear{2026}
\acmYear{2026}
\setcopyright{cc}
\setcctype{by}
\acmConference[MM '26]{Proceedings of the 35th ACM International Conference on Multimedia}{November 10--14, 2026}{Rio de Janeiro, Brazil}
\acmBooktitle{Proceedings of the 35th ACM International Conference on Multimedia (MM '26), November 10--14, 2026, Rio de Janeiro, Brazil}
\acmDOI{10.1145/3767308.3836382}
\acmISBN{979-8-4007-2213-4/2026/11}

\begin{document}

\title[IIMRec]{One Graph, Multiple Gains: Single High-Quality Item-Item Graph for Multimodal Recommendation}


\author{Jinfeng Xu}
\email{jinfeng@connect.hku.hk}
\affiliation{%
  \institution{The University of Hong Kong}
  \city{HongKong SAR}
  \country{China}}

\author{Zheyu Chen}
\email{zheyu.chen@connect.polyu.hk}
\affiliation{%
  \institution{The Hong Kong Polytechnic University}
  \city{HongKong SAR}
  \country{China}}

\author{Ziyue Peng}
\email{pengziyue001@gmail.com}
\affiliation{%
  \institution{The Hong Kong University of Science and Technology}
  \city{HongKong SAR}
  \country{China}}

\author{Shuo Yang}
\email{shuoyang.ee@gmail.com}
\affiliation{%
  \institution{The University of Hong Kong}
  \city{HongKong SAR}
  \country{China}}

\author{Jinze Li}
\email{lijinze-hku@connect.hku.hk}
\affiliation{%
  \institution{The University of Hong Kong}
  \city{HongKong SAR}
  \country{China}}

\author{Zewei Liu}
\email{zewliu@eee.hku.hk}
\affiliation{%
  \institution{The University of Hong Kong}
  \city{HongKong SAR}
  \country{China}}

\author{Shujie Li}
\email{shujie.li@connect.hku.hk}
\affiliation{%
  \institution{The University of Hong Kong}
  \city{HongKong SAR}
  \country{China}}

\author{Yipeng Du}
\email{yipengdu@connect.hku.hk}
\affiliation{%
  \institution{The University of Hong Kong}
  \city{HongKong SAR}
  \country{China}}

\author{Edith C. H. Ngai}
\authornote{Corresponding authors}
\email{chngai@eee.hku.hk}
\affiliation{%
  \institution{The University of Hong Kong}
  \city{HongKong SAR}
  \country{China}}

\renewcommand{\shortauthors}{Xu et al.}

\begin{abstract}
Multimodal recommendation leverages item multimodal features alongside collaborative signals to capture user preferences. While item-item graphs have become a key building block in advanced models, existing methods typically construct them with noisy similarity edges and limit their role to a single function of item-item representation propagation, leaving substantial potential untapped.

In this paper, we propose IIMRec, a framework that constructs a single high-quality item-item graph during preprocessing and systematically reuses it across three stages of the recommendation pipeline: representation enhancement, interaction graph enhancement, and optimization enhancement. The graph is built by fusing semantic and co-occurrence signals, then refined via Neighborhood Consistency Edge Reweighting (NCER), which applies the triadic closure principle to amplify structurally reliable edges and suppress spurious ones. Once constructed, the graph is leveraged in three complementary ways: (1) Item-item propagation with a Residual II Gate (RIG) that adaptively controls per-item absorption of semantic neighborhood signals for representation enhancement; (2) A content-guided UI graph expansion that introduces virtual user-item edges through high-confidence semantic neighbors for interaction graph enhancement; (3) II-Neighbor BPR Augmentation (INA) that treats top neighbors of positive items as discounted soft positives for optimization enhancement.

All components are computed or cached during preprocessing or add negligible per-batch cost, making IIMRec both effective and efficient. We provide theoretical analysis showing that NCER reduces the spectral noise-to-signal ratio, RIG converges to a non-degenerate gating regime, and INA yields a tighter generalization bound. Extensive experiments on four datasets demonstrate that IIMRec consistently outperforms state-of-the-art baselines while running faster and consuming less GPU memory, with particularly strong gains under cold-start and sparse-interaction conditions.
\end{abstract}
\maketitle

\section{Introduction}
\label{sec:Introduction}

Multimodal recommendation jointly exploits diverse item metadata (images, text, audio) and collaborative signals to model user preferences~\cite{chen2025don}. Graph-based methods dominate the current landscape by propagating representations along a user-item bipartite graph and one or more item-item affinity graphs, capturing high-order collaborative patterns and modality-specific semantics simultaneously~\cite{xu2025survey,zhou2023comprehensive,chen2025don,zhou2023bootstrap}. Among these structures, the \emph{item-item homogeneous graph} has emerged as a particularly important building block, enabling items to enhance semantic information through modality-aware neighborhoods~\cite{xu2026well,zhang2021mining,zhou2023tale,xu2025cohesion}.

Despite its central role, we identify two fundamental limitations. \textbf{First, graph quality is compromised.} The prevailing methods construct item-item graphs via top-$k$ cosine similarity in pretrained feature spaces~\cite{zhang2021mining,zhou2023tale,xu2025cohesion}, which is sensitive to hub effects and feature contamination~\cite{wei2020graph,xu2025nlgcl,xu2026nlgclp}. Recent remedies such as co-occurrence graphs~\cite{xu2025cohesion} or spectral filtering~\cite{ong2025spectrum} apply uniform treatment to every edge without evaluating the structural evidence behind each connection. \textbf{Second, graph utilization is narrow.} The vast majority of existing methods exploit the item-item graph exclusively for propagating item representations via a fixed aggregation~\cite{zhang2021mining,chen2025hypercomplex,xu2025cohesion}, while only a handful of recent works have explored using item-item affinity to augment the interaction graph~\cite{xu2025vi,xu2025mdvt}. Furthermore, the fixed aggregation universally adopted in these methods treats the item-item propagation signal as equally beneficial for every item, overlooking recent findings that graph convolution can harm representations of items whose modality features are noisy or inconsistent with their neighborhood~\cite{xu2025best,chen2025don}.

Our key insight is that \emph{a single high-quality item-item graph, constructed once during preprocessing, can serve as a versatile foundation across multiple pipeline stages}. Based on this insight, we propose \textbf{\method{}}, a single high-quality \textbf{I}tem-\textbf{I}tem Graph for \textbf{M}ultimodal \textbf{Rec}ommendation. \method{} first refines a fused semantic-and-co-occurrence item-item graph through \textbf{Neighborhood Consistency Edge Reweighting (NCER)}, which amplifies edges with high shared-neighbor counts and suppresses isolated ones, provably reducing the spectral noise-to-signal ratio (Theorem~\ref{thm:ncer}). The refined graph is then systematically leveraged in three complementary ways. For \emph{representation enhancement}, a \textbf{Residual II Gate (RIG)} replaces the fixed additive aggregation with a differentiable per-item gate that controls how much item-item propagation signal each item absorbs, with provably non-degenerate convergence (Theorem~\ref{thm:rig}). For \emph{interaction graph enhancement}, a content-guided UI expansion introduces virtual user-item edges through high-confidence semantic neighbors, benefiting cold-start and sparse-interaction conditions. For \emph{optimization enhancement}, \textbf{II-Neighbor BPR Augmentation (INA)} treats top-$k$ neighbors of each positive item as discounted soft positives that push in the same direction as BPR, yielding a tighter generalization bound (Theorem~\ref{thm:ina}). All components are pre-computed or cached at negligible per-batch cost.

Our main contributions are as follows:
\begin{itemize}[leftmargin=*]
    \item We identify the dual limitations of low graph quality and narrow utilization in existing item-item graphs, and propose a "\emph{build once, leverage everywhere}" paradigm that reuses a single refined graph for representation, interaction, and optimization.
    \item We design three complementary components (NCER, RIG, INA), each backed by theoretical guarantees on noise reduction, non-degeneracy, and generalization, respectively.
    \item Extensive experiments on four datasets show that \method{} consistently outperforms advanced baselines with lower training time and memory consumption, achieving particularly strong gains under cold-start and sparse settings.
\end{itemize}
\section{Related Work}
\label{sec:Related work}
Multimodal recommendation exploits diverse item multimodal metadata alongside collaborative signals to improve preference modeling. Early work \cite{he2016vbpr,tang2019adversarial,li2025ddunet} injected visual features into matrix factorization. Subsequent graph-based methods introduced bipartite graphs for user-item modeling \cite{wei2019mmgcn,wei2020graph} and item-item homogeneous graphs for capturing modality-specific relationships \cite{zhou2023tale,zhang2021mining}. More recently, SMORE \cite{ong2025spectrum} applied spectral filtering to behavioral modalities, HPMRec \cite{chen2025hypercomplex} explored hyper-complex representations, COHESION \cite{xu2025cohesion} combined dual-stage fusion with composite graph convolution, and several studies \cite{zhou2023bootstrap,xu2024mentor,jiang2024diffmm} adopted self-supervised learning to alleviate data sparsity.

Despite these advances, two limitations persist in how existing methods handle item-item graphs. On the construction side, graphs built from top-$k$ cosine similarity \cite{zhang2021mining,zhou2023tale} or co-occurrence counts \cite{xu2025cohesion} inevitably contain noisy edges caused by hub effects and feature contamination \cite{xu2025nlgcl,wei2020graph}; remedies such as spectral filtering \cite{ong2025spectrum} and graph freezing \cite{zhou2023tale} treat every edge uniformly without assessing its structural reliability. On the utilization side, the vast majority of methods exploit the item-item graph exclusively for propagating item representations via a fixed additive residual \cite{xu2025survey,chen2025hypercomplex,zhang2021mining}, leaving its potential for enriching interaction structures \cite{xu2025vi} or augmenting training supervision untapped. The fixed aggregation universally adopted further treats the propagation signal as equally beneficial for every item, overlooking findings that graph convolution can produce harmful positive-negative pairs \cite{xu2025best} and that node-neighbor discrepancy can degrade representations \cite{chen2025don}. Multimodal alignment \cite{xu2024mentor,jiang2024diffmm} partially enriches supervision but risks conflicting with ranking by forcing collaborative and semantic spaces to coincide.

IIMRec addresses both limitations under a unified "\emph{build once, leverage everywhere}" paradigm. A single high-quality item-item graph, refined by NCER via triadic closure, is systematically reused across three pipeline stages: representation enhancement through RIG's adaptive per-item gating, interaction graph enhancement via content-guided UI expansion, and optimization enhancement through INA's discounted soft positives. All components are lightweight, modular, and theoretically grounded.
\section{Preliminary}
\label{sec:Preliminary}
Let $\cU$ and $\cI$ denote the sets of users and items, respectively. The observed implicit feedback is represented by a binary interaction matrix $\bm{R} \in \{0,1\}^{|\cU|\times|\cI|}$, where $\bm{R}_{u,i} = 1$ indicates that user $u$ has interacted with item $i$. We construct a user-item bipartite graph $\cG = (\cU \cup \cI, \cE)$ with edge set $\cE = \{(u,i) | \bm{R}_{u,i} = 1\}$. Each item $i \in \cI$ carries modality-specific raw features $\bx^m_i$ for each modality $m \in \cM$ (e.g., visual, textual, and audio). These features are extracted by pre-trained encoders to obtain modality-specific feature matrices $\bm{F}^m \in \R^{|\cI| \times d_m}$, where $d_m$ denotes the feature dimensionality of modality $m$. The goal is to learn a scoring function $f(\cdot)$ that predicts the preference score $\hat{y}_{u,i}$ of user $u$ for item $i$, enabling accurate ranking of candidate items according to user preferences.

\begin{figure*}
    \centering
    \includegraphics[width=1\linewidth]{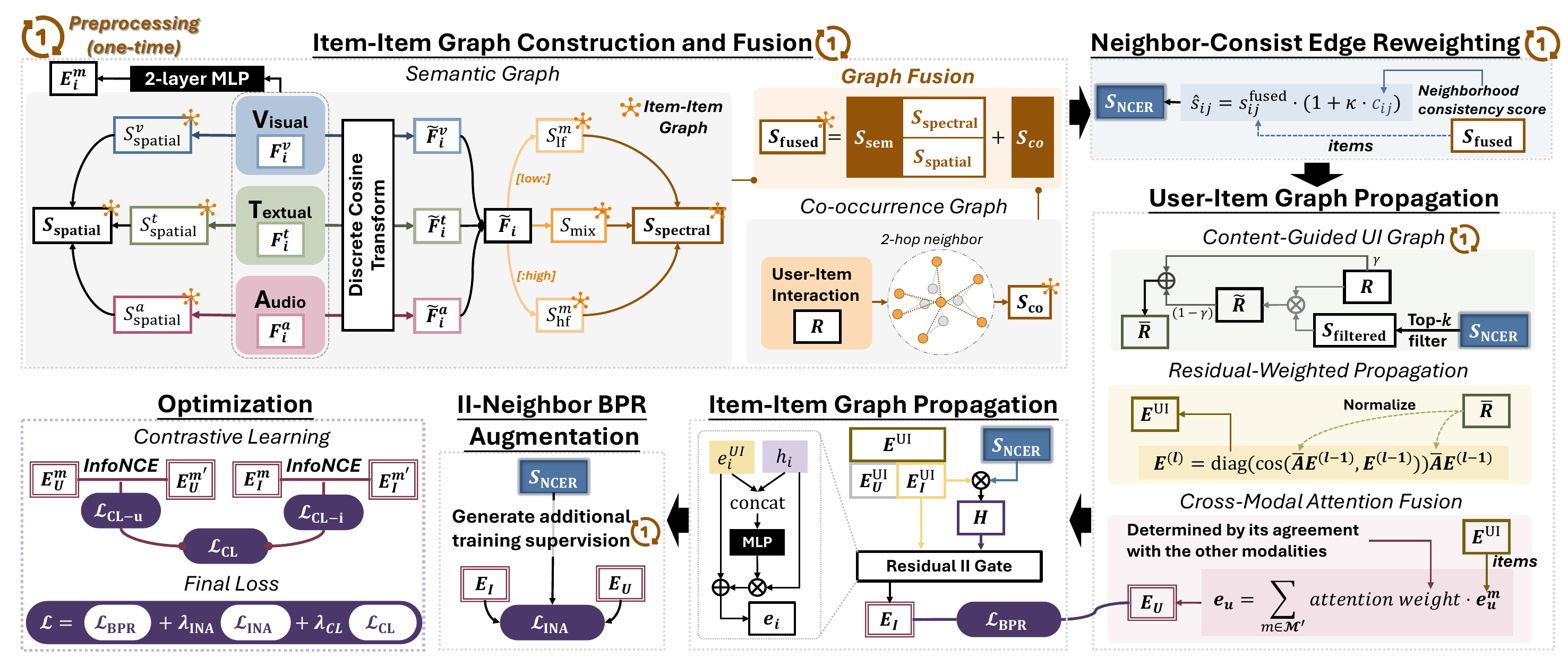}
    \vskip -0.15in
    \caption{The overall structure of IIMRec.}
    \label{fig:framework}
    \vskip -0.1in
\end{figure*}

\section{Methodology}
\label{sec:Method}

We present \method{}\footnote{Code can be found in: \href{https://github.com/Jinfeng-Xu/IIMRec}{https://github.com/Jinfeng-Xu/IIMRec}.}, a single high-quality \textbf{I}tem-\textbf{I}tem Graph for \textbf{M}ultimodal \textbf{Rec}ommendation. As illustrated in Figure~\ref{fig:framework}, we detail each component below. To ensure uniformity in the semantic space and maintain model generalizability, we use the same hyper-parameter $k$ for all $\operatorname{Top\text{-}k}(\cdot)$ operations and conduct a unified hyper-parameter sensitivity analysis in Section~\ref{sec:hyper}.


\subsection{Multimodal Feature Projection}
\label{sec:feature_proj}

To ensure dimensional consistency across heterogeneous modality features, each modality's raw features are projected via a two-layer MLP into a unified embedding space of dimension $d$:
\begin{equation}
\label{eq:projection}
\be^m_i = \bW^m_2 \text{LeakyReLU}(\bW^m_1 \bm{f}^m_i), \quad \forall m \in \cM,
\end{equation}
where $\bW^m_1 \in \R^{4d \times d_m}$ and $\bW^m_2 \in \R^{d \times 4d}$ are learnable parameter matrices, with the intermediate dimension $4d$ providing sufficient capacity for cross-dimensional mapping. Since users lack multimodal raw features, user embeddings $\bE^m_u \in \R^{|\cU| \times d}$ are randomly initialized for each modality $m$. 

Beyond the spatial (original) features, we extract \emph{spectral features} via the Discrete Cosine Transform (DCT) to capture frequency-domain patterns that are complementary to spatial representations:
\begin{equation}
\label{eq:dct}
\tilde{\bm{f}}_i = \operatorname{Concat}(\text{DCT}(\bm{f}^m_i) | m \in \cM),
\end{equation}
where $\operatorname{Concat}(\cdot | \cdot)$ denotes concatenation across all modalities and $\text{DCT}(\cdot)$ applies the orthonormal Discrete Cosine Transform. The resulting spectral features $\tilde{\bm{f}}_i \in \R^{\sum_m d_m}$ encode global structure (in low-frequency components) and fine-grained local patterns (in high-frequency components). These spectral features are projected through a separate two-layer MLP with input dimension $\sum_m d_m$ and identical hidden/output dimensions as in Eq.~\eqref{eq:projection}, yielding spectral embeddings $\be^s_i \in \R^d$. For users, a separate learnable spectral preference embedding $\be^s_u \in \R^d$ is maintained.

The initial node embedding matrix $\bE^{(0)} \in \R^{(|\cU|+|\cI|) \times ((|\cM|+1)d)}$ is formed by concatenating the per-modality node embeddings along the feature dimension. For each node $n \in \cU \cup \cI$:
\begin{equation}
\label{eq:init_embed}
\be^{(0)}_n = \operatorname{Concat}(\operatorname{Norm}(\be^s_n),\; \operatorname{Norm}(\be^m_n)|m \in \cM),
\end{equation}
where $\operatorname{Norm}(\cdot)$ denotes $\ell_2$ normalization, which ensures that all modality channels contribute equally to the initial representation regardless of their raw feature magnitudes. Here, $\be^m_n$ denotes the projected item embedding from Eq.~\eqref{eq:projection} if $n \in \cI$, or the learnable user preference embedding $\be^m_n$ if $n \in \cU$.

\subsection{Item-Item Graph Construction and Fusion}
\label{sec:ii_graph}

The item-item affinity graph serves as the second propagation channel and is constructed by fusing two complementary sources that capture different aspects of item relatedness. Note that the item–item graph is constructed before training and incurs no computational overhead during the training phase.

\textbf{Semantic Graph.}
The semantic graph $\bm{S}_{\text{sem}}$ captures modality-level feature similarity and is itself a blend of spatial and spectral components. For the spatial component, we construct a KNN graph separately for each modality by computing pairwise cosine similarity among all item pairs and retaining the top-$k$ neighbors per item:
\begin{equation}
\label{eq:sem_knn}
\bm{S}^m_{\text{spatial}} = \operatorname{RowNorm}\!\left(\operatorname{Top\text{-}k}\!\left(\frac{\bm{F}^m (\bm{F}^m)^\top}{\|\bm{F}^m\|_{\text{row}} \|\bm{F}^m\|_{\text{row}}^\top}\right)\right), \quad \forall m \in \cM,
\end{equation}
where $\operatorname{Top\text{-}k}(\cdot)$ zeros out all entries except the $k$ largest per row, and $\operatorname{RowNorm}(\cdot)$ normalizes each row to sum to one. The modality-specific spatial graphs are then fused: $\bm{S}_{\text{spatial}} = \sum_{m \in \cM} \bm{S}^m_{\text{spatial}}$.

For each modality $m$, we compute $\bm{\tilde{F}}^m = \operatorname{DCT}(\bm{F}^m)$ and partition its columns into low-frequency ($\bm{\tilde{F}}^m_{\mathrm{lf}}$, first $\lfloor \rho \cdot d_m \rfloor$ coefficients) and high-frequency ($\bm{\tilde{F}}^m_{\mathrm{hf}}$, remainder) bands with $\rho=0.5$. From each band of each modality, we construct a KNN graph as in Eq.~\eqref{eq:sem_knn}, yielding $\bm{S}^m_{\mathrm{lf}}$ and $\bm{S}^m_{\mathrm{hf}}$. We additionally build a cross-modal mixed-frequency graph $\bm{S}_{\mathrm{mix}}$ from concatenated DCT features $\bm{\tilde{F}}_{\mathrm{all}} = \mathrm{Concat}(\bm{\tilde{F}}^m | m \in \mathcal{M})$. The spectral semantic graph is $\bm{S}_{\mathrm{spectral}} = \sum_m \bm{S}^m_{\mathrm{lf}} + \sum_m \bm{S}^m_{\mathrm{hf}} + \bm{S}_{\mathrm{mix}}$, and the complete semantic graph is $\bm{S}_{\mathrm{sem}} = \bm{S}_{\mathrm{spatial}} + \bm{S}_{\mathrm{spectral}}$.

\textbf{Co-occurrence Graph.}
The co-occurrence graph $\bm{S}_{\text{co}}$ captures behavioral proximity from shared user interactions, providing collaborative evidence that is complementary to the content-based semantic graph. We compute the item co-occurrence matrix $\bm{C} = \bm{R}^\top \bm{R}$, where $C_{ij}$ counts the number of users who interacted with both items $i$ and $j$. To control density, we set the diagonal to zero and apply $\operatorname{Top\text{-}k}(\cdot)$ filtering per item. The filtered co-occurrence counts are then normalized via softmax to produce probability-like edge weights:
\begin{equation}
\label{eq:co_graph}
s^{\text{co}}_{ij} = \frac{\exp(C_{ij})}{\sum_{j' \in \operatorname{Top\text{-}k}(i)} \exp(C_{ij'})}.
\end{equation}

Therefore, we obtain the co-occurrence graph $\bm{S}_{\text{co}}$.


\textbf{Graph Fusion.}
The semantic and co-occurrence graphs are fused:
\begin{equation}
\label{eq:graph_fusion}
\bm{S}_{\text{fused}} = \bm{S}_{\text{co}} + \bm{S}_{\text{sem}},
\end{equation}
This fused graph then serves as the input to the NCER reweighting.

\subsection{Neighborhood Consistency Edge Reweighting (NCER)}
\label{sec:ncer}

The fused item-item graph $\bm{S}_{\text{fused}}$ inevitably contains noisy edges arising from spurious feature similarities or co-occurrence artifacts. We propose NCER, which leverages the principle of \emph{triadic closure} to upweight structurally reliable edges and suppress noisy ones.

\textbf{Motivation.} In social network theory, an edge between two nodes is considered more trustworthy if the nodes share many common neighbors~\cite{easley2010networks}. Applied to item-item graphs: if items $i$ and $j$ are both similar to a large overlapping set of other items, the edge $(i,j)$ is supported by multiple independent paths of evidence. Conversely, if $i$ and $j$ share no neighbors, the edge is an isolated connection more likely to be a noise artifact.

\textbf{Mechanism.} We first convert the fused graph into a binary adjacency matrix $\bm{B} \in \{0,1\}^{|\cI| \times |\cI|}$ where $\bm{B}_{ij} = \mathbb{I}[s^{\text{fused}}_{ij} > 0]$. The shared neighbor count for each edge is then efficiently computed via matrix multiplication: $(\bm{B}\bm{B}^\top)_{ij} = |\cN(i) \cap \cN(j)|$, where $\cN(i)$ denotes the neighbor set of item $i$. We normalize this to obtain the consistency score:
\begin{equation}
\label{eq:ncer_overlap}
c_{ij} = \frac{|\cN(i) \cap \cN(j)|}{k} = \frac{(\bm{B}\bm{B}^\top)_{ij}}{k},
\end{equation}
where $k$ is the KNN parameter and $c_{ij} \in [0,1]$ quantifies the fraction of shared neighbors. Each edge is then reweighted proportionally to its consistency:
\begin{equation}
\label{eq:ncer_reweight}
\hat{s}_{ij} = s^{\text{fused}}_{ij} \cdot (1 + \kappa \cdot c_{ij}),
\end{equation}
where $\kappa > 0$ controls the reweighting strength. The reweighted graph is row-normalized to obtain the final II adjacency $\bm{S}_{\text{NCER}}$ via $\hat{s}_{ij}$. Notably, NCER is computed once during preprocessing and adds no cost during training or inference.

\textbf{Connection to Graph Laplacian.} The reweighting in Eq.~\eqref{eq:ncer_reweight} can be interpreted as preferentially preserving the low-rank component of the adjacency matrix. In Section~\ref{sec:theory}, we formalize this by showing that NCER reduces the spectral norm of the noise component $\bm{N}$ in the decomposition $\bm{S} = \bm{S}^* + \bm{N}$, where $\bm{S}^*$ is the ideal (clean) graph.

\subsection{User-Item Graph Propagation}
\label{sec:ui_propagation}

To capture high-order collaborative signals, we propagate node representations on the user-item bipartite graph. We adopt LightGCN~\cite{he2020lightgcn} as the message-passing backbone, which removes feature transformation and nonlinear activation from standard GCNs. Specifically, we compute the symmetric normalized adjacency matrix $\tilde{\bm{A}} = \bm{D}^{-1/2} \bm{A} \bm{D}^{-1/2}$ from the bipartite adjacency $\bm{A} \in \R^{(|\cU|+|\cI|) \times (|\cU|+|\cI|)}$, where $\bm{D}$ is the diagonal degree matrix.

\textbf{Content-Guided UI Graph.}
We expand the interaction structure using the NCER-refined graph. Applying a joint Top-$k$ filter to $\bm{S}_{\mathrm{NCER}}$ yields a sparse subgraph $\bm{S}_{\mathrm{filtered}}$. The expanded interaction matrix $\tilde{\bm{R}} = \bm{R}\bm{S}_{\mathrm{filtered}}$ creates virtual user-item edges through high-confidence semantic neighbors, and is blended with the original: $\bar{\bm{R}} = \gamma \bm{R} + (1-\gamma)\tilde{\bm{R}}$. The symmetric normalized adjacency $\bar{\bm{A}}$ of $\bar{\bm{R}}$ replaces $\tilde{\bm{A}}$ during training.

\textbf{Residual-Weighted Propagation.}
Multi-layer propagation enriches representations with high-order collaborative signals, but deeper layers also risk injecting noise from distant and less relevant neighbors \cite{chen2025don}. To address this, we employ a residual weighting scheme that adaptively attenuates each propagation layer based on its agreement with the previous-layer representation:
\begin{equation}
\label{eq:residual_prop}
\bE^{(l)} = \operatorname{diag}\!\left(\cos(\bar{\bm{A}}\bE^{(l-1)},\; \bE^{(l-1)})\right) \bar{\bm{A}}\bE^{(l-1)},
\end{equation}
where $\cos(\cdot,\cdot)$ computes the row-wise cosine similarity between the propagated output and the input from the previous layer, and $\operatorname{diag}(\cdot)$ constructs a diagonal matrix from these per-node similarity scores. This mechanism rescales the propagated embedding of each node by how much it agrees with the node's current representation: nodes whose neighborhoods reinforce the existing signal receive full propagation strength, while nodes whose neighborhoods introduce divergent information are attenuated. The final UI-propagated representation aggregates across all layers: $\bE^{\text{UI}} = \sum_{l=0}^{L} \bE^{(l)}$, where $L$ is the number of propagation layers. The item representations $\be^{\text{UI}}_i$ and the per-modality user representations are then extracted from the corresponding rows of $\bE^{\text{UI}}$.


\textbf{Cross-Modal Attention Fusion.}
After UI-graph propagation, the user portion of $\bE^{\text{UI}}$ contains $|\cM|+ 1$ modality channels that must be fused into a unified user embedding. Rather than using a fixed or heuristic fusion strategy \cite{zhou2023tale,xu2025enhancing,xu2025mdvt}, we adopt a cross-modal attention mechanism where the weight of each modality is determined by its agreement with the other modalities. Each modality's user representation is first projected and $\ell_2$-normalized via a modality-specific linear transformation: $\bz^m_u = \operatorname{Norm}(\bW^m_{\text{proj}} \be^m_u)$, where $\bW^m_{\text{proj}} \in \R^{d \times d}$ is a learnable projection. The attention weight for modality $m$ is computed as:
\begin{equation}
\label{eq:user_fusion}
\be_u = \sum_{m \in \cM'} \frac{\exp\!\left(\sum_{m' \neq m} \operatorname{sim}(\bz^m_u, \bz^{m'}_u) / \tau\right)}{\sum_{m'' \in \cM'} \exp\!\left(\sum_{m' \neq m''} \operatorname{sim}(\bz^{m''}_u, \bz^{m'}_u) / \tau\right)} \be^m_u,
\end{equation}
where $\cM' = \{s\} \cup \cM$ denotes the set of modality channels, $\operatorname{sim}(\cdot,\cdot)$ denotes cosine similarity, and $\tau$ is a temperature parameter (We fixed $\tau = 0.2$ in practice). The intuition is that a modality which agrees well with the other modalities is more likely to carry reliable preference information and should receive higher weight.


\subsection{Item-Item Graph Propagation via Residual II Gate (RIG)}
\label{sec:rig}

Given the NCER-refined item-item adjacency $\bm{S}_{\text{NCER}}$, we propagate item representations through a single layer of message passing on the item-item graph to enrich each item's representation with semantic neighborhood information:
\begin{equation}
\label{eq:ii_prop}
\bH = \bm{S}_{\text{NCER}} \bE^{\text{UI}}_{\cI},
\end{equation}
where $\bE^{\text{UI}}_{\cI} \in \R^{|\cI| \times (|\cM|+1)d}$ denotes the item portion of the UI-propagated representations. The propagated output $\bh_i$ for each item $i$ encodes aggregated semantic signals from its neighborhood in the item-item graph.

\textbf{Residual II Gate (RIG).}
The standard approach adds $\bh_i$ to $\be^{\text{UI}}_i$ via a fixed residual: $\be_i = \be^{\text{UI}}_i + \bh_i$. This uniform treatment is suboptimal: for items with reliable modality features, $\bh_i$ provides valuable semantic enrichment, while for items with noisy or contradictory features, $\bh_i$ may corrupt the collaborative signal learned from UI propagation. We therefore replace the fixed residual with a learned per-item gate that adaptively controls each item's II contribution:
\begin{equation}
\label{eq:rig}
\be_i = \be^{\text{UI}}_i +  \sigma\left(\bW_2 \cdot \text{ReLU}(\bW_1 \cdot \operatorname{Concat}(\be^{\text{UI}}_i ,\bh_i)) \right) \bh_i,
\end{equation}
where $\bW_1 \in \R^{(|\cM|+1)d \times 2(|\cM|+1)d}$ and $\bW_2 \in \R^{1 \times (|\cM|+1)d}$ are learnable parameters matrices, $\operatorname{Concat}(\cdot,\cdot)$ denotes concatenation, and $\sigma$ denotes the sigmoid function. The gate takes as input both the UI-derived representation (encoding collaborative signals) and the II-propagated output (encoding semantic neighborhood signals), and produces a scalar in $[0,1]$ that controls the II contribution for each item individually.

Because the gate parameters are optimized end-to-end through the BPR loss, the model naturally learns to assign high gate values to items whose II-propagated signals improve recommendation quality (i.e., items with consistent, high-quality modality features where II propagation aggregates coherent signals from genuine semantic neighbors) and low gate values to items whose II signals would degrade it (i.e., items with noisy features where II propagation mixes contradictory or irrelevant signals).

\textbf{Design Rationale.} Unlike attention-based modality selection~\cite{kim2024monet}, which operates across modalities for a fixed item, RIG operates \emph{across items} for a fixed propagation path. The gate addresses the question of how much the II-propagated semantic signal should contribute to each specific item's final representation.

\subsection{II-Neighbor BPR Augmentation (INA)}
\label{sec:ina}

Standard BPR~\cite{rendle2012bpr} training constructs each training triplet $(u, p, n)$ by sampling a positive item $p$ that user $u$ has interacted with and a negative item $n$ randomly drawn from the set of non-interacted items. The supervision signal comes exclusively from the interaction matrix, which is sparse in practice. We propose to enrich this signal by leveraging the II graph to generate additional training supervision.

\textbf{Motivation.} 
If item $j$ is a top neighbor of positive item $p$ in the II graph, $j$ shares strong semantic similarity with a user-preferred item and can serve as a \emph{soft positive}, i.e., an item the user likely prefers over a random negative but with lower confidence than the observed interaction. This is empirically supported by~\cite{xu2025vi}, which shows that modality-similar items of interacted items overlap significantly more with true preferences than random items.


\textbf{Mechanism.} For each item $i$, we pre-compute its top-$k$ neighbors from the NCER-refined II graph by selecting the $k$ neighbors with the highest edge weights in $\bm{S}_{\text{NCER}}$. These are stored in a lookup table $\bm{T} \in \mathbb{R}^{|\cI| \times k}$, where $T_{i,k}$ is the index of the $k$-th nearest neighbor. During training, for each positive item $p$ in a mini-batch, we randomly sample one valid neighbor $j$ from $\bm{T}_p$ and independently sample a fresh random negative $n'$, then compute a discounted BPR loss:
\begin{equation}
\label{eq:ina_loss}
\cL_{\text{INA}} = \frac{-\delta}{|\cB_{\text{valid}}|} \sum_{(u,p,n') \in \cB_{\text{valid}}} \log \sigma(\be_u^\top \be_j - \be_u^\top \be_{n'}),
\end{equation}
where $\delta \in (0,1)$ is a discount factor reflecting the reduced confidence of soft positives compared to observed interactions, and $\cB_{\text{valid}}$ denotes the subset of the mini-batch where the positive item has at least one valid II-graph neighbor. The loss is computed only over this valid subset to avoid corrupting the gradient with items that lack meaningful semantic neighbors.

\textbf{Compatibility with BPR.} A critical distinction from cross-graph alignment approaches \cite{xu2024mentor,zhou2023bootstrap,wei2023multi} (which force UI-propagated and II-propagated representations to align via contrastive loss) is that INA pushes in the \emph{same direction} as BPR: both aim to make the user embedding closer to desirable items than to random negatives. Cross-graph alignment, by contrast, can conflict with BPR because UI representations encode collaborative filtering signals while II representations encode semantic similarity; forcing alignment between these can destroy the CF signal that drives recommendation quality.

\subsection{Training Objective and Complexity}
\label{sec:training}

\textbf{BPR Loss.} The primary recommendation objective is the Bayesian Personalized Ranking (BPR) loss~\cite{rendle2012bpr}, which encourages the model to rank positive items above negative items for each user:
\begin{equation}
\label{eq:bpr}
\cL_{\text{BPR}} = \sum_{(u,p,n) \in \cD} -\log \sigma\!\left(\be_u^\top \be_p - \be_u^\top \be_n\right),
\end{equation}
where $\cD$ is the set of training triplets, $\be_u$ is the fused user representation from Eq.~\eqref{eq:user_fusion}, and $\be_p, \be_n$ are the RIG-gated item representations from Eq.~\eqref{eq:rig}.

\textbf{Contrastive Learning Losses.} 
We employ two contrastive learning objectives to encourage cross-modal alignment. The user-level contrastive loss aligns projected user representations across modality pairs using the InfoNCE~\cite{oord2018representation} formulation:
\begin{equation}
\label{eq:cl_user}
\cL_{\text{CL-u}} = \sum_{m,m' \in \cM} \frac{-1}{|\cB|} \sum_{u \in \cB} \log \frac{\exp(\bz^m_u \cdot \bz^{m'}_u / \tau_{\text{cl}})}{\sum_{u'} \exp(\bz^m_u \cdot \bz^{m'}_{u'} / \tau_{\text{cl}})},
\end{equation}
where $\bz^m_u$ is the projected and normalized user representation (Eq.~\ref{eq:user_fusion}), $\cB$ is the current mini-batch, and $\tau_{\text{cl}}$ is a contrastive temperature (We fixed $\tau_{\text{cl}} = 0.2$ in practice). The item-level contrastive loss $\cL_{\text{CL-i}}$ is defined analogously over the modality-specific item representations after UI propagation, using the same modality pairs and temperature. We obtain the $\cL_{\text{CL}} = \cL_{\text{CL-u}} + \cL_{\text{CL-i}}$.


\textbf{Total Loss.} The complete training loss combines all losses:
\begin{equation}
\label{eq:total_loss}
\cL = \cL_{\text{BPR}} + \lambda_{\text{INA}} \cL_{\text{INA}} + \lambda_{\text{CL}} \cL_{\text{CL}},
\end{equation}
where $\lambda_{\text{INA}}$ and $\lambda_{\text{CL}}$ are the respective loss weights.

\textbf{Prediction.} At inference time, the predicted preference score for user $u$ on item $i$ is computed as the inner product of their final representations: $\hat{y}_{u,i} = \be_u^\top \be_i$. All items are then ranked by their predicted scores for each user to get the top-$N$ recommendation list.

The complete training procedure is summarized in Algorithm~\ref{alg:iimrec} in Appendix~\ref{appendix:alg}.

\textbf{Complexity Analysis.} NCER adds $O(|\cI| \cdot k^2)$ preprocessing cost for computing shared neighbor counts via $\bm{B}\bm{B}^\top$ (where $k$ is the KNN parameter), which is negligible relative to KNN graph construction. RIG adds $O((d(|\cM|+1)^2)$ parameters for the two-layer MLP gate, with per-batch forward cost $O(B \cdot (d(|\cM|+1))$ where $B$ is the batch size. INA adds $O(B)$ per-batch cost for neighbor sampling and a discounted BPR computation. All three innovations are lightweight relative to the GCN propagation $O(|\cE| \cdot (d(|\cM|+1))$ that dominates the overall training cost. Notably, NCER and the INA neighbor lookup table are both computed once during preprocessing and cached to disk, incurring zero additional cost during training. We provide an empirical efficiency analysis (See Section~\ref{sec:efficiency}) demonstrating that IIMRec significantly outperforms existing advanced multimodal recommendation models in terms of efficiency.

\section{Theoretical Analysis}
\label{sec:theory}
We provide theoretical guarantees for IIMRec's three innovations.
 
\subsection{NCER Reduces Graph Spectral Noise}
\label{sec:theory_ncer}
 
\begin{assumption}[Noisy Graph Model]
\label{ass:noisy_graph}
The observed item-item similarity matrix decomposes as $\bm{S} = \bm{S}^* + \bm{N}$, where $\bm{S}^*$ is the ideal (clean) graph encoding true semantic relationships and $\bm{N}$ is a noise matrix with $\|\bm{N}\|_2 \leq \epsilon$. We assume $\bm{S}^*$ exhibits community structure such that true edges have high neighborhood overlap while noise edges have low overlap.
\end{assumption}
 
\begin{theorem}[NCER Noise Reduction]
\label{thm:ncer}
Under Assumption~\ref{ass:noisy_graph}, let $\hat{\bm{S}} = \bm{S} \circ \bm{W}$ denote the NCER-reweighted graph where $W_{ij} = 1 + \kappa \cdot c_{ij}$ and $\circ$ denotes element-wise multiplication. Decompose $\hat{\bm{S}} = \hat{\bm{S}}^* + \hat{\bm{N}}$ correspondingly. Then:
\begin{equation}
\frac{\|\hat{\bm{N}}\|_F}{\|\hat{\bm{S}}\|_F} \leq \frac{\|\bm{N}\|_F}{\|\bm{S}\|_F} \cdot \frac{1 + \kappa \bar{c}_N}{1 + \kappa \bar{c}_{S^*}},
\end{equation}
where $\bar{c}_N$ and $\bar{c}_{S^*}$ are the average neighborhood consistency scores of noise edges and clean edges, respectively. When $\bar{c}_{S^*} > \bar{c}_N$ (true edges have higher overlap than noise edges), NCER strictly reduces the noise-to-signal ratio: $\|\hat{\bm{N}}\|_F / \|\hat{\bm{S}}\|_F < \|\bm{N}\|_F / \|\bm{S}\|_F$.
\end{theorem}
 
The proof (Appendix~\ref{app:proof_ncer}) follows from the observation that NCER multiplicatively amplifies edges proportionally to their consistency scores, and under the community structure assumption, clean edges systematically have higher scores than noise edges.
 
\subsection{RIG Convergence and Non-Degeneracy}
\label{sec:theory_rig}
 
A potential concern with learned gates is collapse: all gates converging to 0 (ignoring II) or 1 (ignoring UI).
 
\begin{theorem}[RIG Non-Degeneracy]
\label{thm:rig}
Consider the BPR loss $\cL = -\sum_{(u,p,n)} \log \sigma(\be_u^\top \be_p - \be_u^\top \be_n)$ with RIG-gated item representations $\be_i = \be^{\text{UI}}_i + g_i \cdot \bh_i$. Let $\theta_g$ denote the gate parameters. Under standard regularity conditions (bounded embeddings, non-degenerate initialization), the gradient of $\cL$ with respect to the gate output $g_i$ satisfies:
\begin{equation}
\frac{\partial \cL}{\partial g_i} = -\sum_{u: (u,i,\cdot) \in \cD} w_{u,i} \cdot \be_u^\top \bh_i + \sum_{u: (u,\cdot,i) \in \cD} w'_{u,i} \cdot \be_u^\top \bh_i,
\end{equation}
where $w_{u,i}, w'_{u,i} > 0$ are instance-dependent weights. When item $i$ is a positive item and $\be_u^\top \bh_i > 0$ (II signal is helpful), the gradient pushes $g_i$ upward. When $\be_u^\top \bh_i < 0$ (II signal is harmful), the gradient pushes $g_i$ downward. Therefore, at any stationary point, $g_i \in (0,1)$ for items that appear as both positives and negatives.
\end{theorem}
 
\begin{corollary}
If there exist items $i,j$ such that the II signal is helpful for $i$ ($\mathbb{E}[\be_u^\top \bh_i] > 0$) and harmful for $j$ ($\mathbb{E}[\be_u^\top \bh_j] < 0$), then at any local minimum, $g_i > g_j$, confirming that the gate differentiates between items based on II signal quality.
\end{corollary}
 
\subsection{INA Tightens Generalization}
\label{sec:theory_ina}
 
\begin{theorem}[INA Generalization Bound]
\label{thm:ina}
Let $\cL_{\text{BPR}}$ and $\cL_{\text{INA}}$ be the standard and augmented BPR losses, with $n$ observed interactions and $n_{\text{aug}}$ augmented soft-positive interactions. Under a PAC-Bayes framework with hypothesis class $\cH$ of bounded-norm embeddings, the expected risk satisfies:
\begin{equation}
R(\hat{f}) \leq \hat{R}_{n+n_{\text{aug}}}(\hat{f}) + O\left(\sqrt{\frac{\text{KL}(\hat{f} \| \pi) + \log(1/\delta)}{n + \delta_{\text{eff}} \cdot n_{\text{aug}}}}\right),
\end{equation}
where $\delta_{\text{eff}} = \delta^2 \cdot \mathbb{P}[j \in \cN_{\text{II}}(p) \Rightarrow \text{user prefers} j]$ captures the effective quality of the augmented signal, and $\hat{R}_{n+n_{\text{aug}}}$ is the empirical risk over all training pairs. When the II graph has positive predictive value ($\delta_{\text{eff}} > 0$), the bound is strictly tighter than the standard BPR bound (which has denominator $n$ only).
\end{theorem}
 
The proof (Appendix~\ref{app:proof_ina}) adapts the standard PAC-Bayes bound to account for the augmented training set, with the discount factor $\delta$ appropriately scaling the contribution of soft positives.


\begin{table}[!ht]
  \centering
  \caption{Statistics of all datasets with multimodal item contents.}
  \vskip -0.15in
  \label{tab:dataset_statistics}
  \resizebox{\linewidth}{!}{
  \begin{tabular}{c c c c c c}
    \toprule
    \textbf{Dataset} & \textbf{Modalities (Dim.)} & \textbf{\#Users} & \textbf{\#Items} & \textbf{\#Interactions} & \textbf{Sparsity} \\
    \midrule
    Baby     & V(4096), T(384)      & 19,445   & 7,050    & 160,792   & 99.88\% \\
    Sports   & V(4096), T(384)      & 35,598   & 18,357   & 296,337   & 99.95\% \\
    Clothing & V(4096), T(384)      & 39,387   & 23,033   & 278,677   & 99.97\% \\
    TikTok   & V(128), T(768), A(128) & 9,319    & 6,710    & 59,541    & 99.90\% \\
    \bottomrule
  \end{tabular}
  }
  \vskip -0.15in
\end{table}

\begin{table*}[!t]
\caption{Recommendation performance comparison of different models. The best and the second-best performances are marked with \textbf{bold} and \underline{underline}, respectively. $^*$ indicates that the t-tests validate the significance of improvements with $p$-value $<$ 0.05.}
\vskip -0.15in
\centering
\tabcolsep=0.04in
\label{tab:result}
\resizebox{\linewidth}{!}{
    \begin{tabular}{ccccccccccccccccc}
     \toprule
         \textbf{Datasets}&  \multicolumn{4}{c}{\textbf{Baby}}&  \multicolumn{4}{c}{\textbf{Sports}}&  \multicolumn{4}{c}{\textbf{Clothing}}&  \multicolumn{4}{c}{\textbf{TikTok}}\\\midrule
         \multicolumn{1}{c}{\textbf{Metrics}}& R@10& R@20& N@10& N@20& R@10& R@20& N@10& N@20& R@10& R@20& N@10& N@20& R@10& R@20& N@10& N@20\\\midrule
         \textbf{LightGCN} & 0.0479& 0.0754& 0.0257& 0.0328& 0.0569& 0.0864& 0.0311& 0.0387& 0.0340& 0.0526& 0.0188& 0.0236& 0.0482& 0.0801& 0.0218& 0.0306\\
         \textbf{SimGCL} & 0.0513 & 0.0804 & 0.0273 & 0.0350& 0.0601 & 0.0919 & 0.0327 & 0.0414 & 0.0356 & 0.0549 & 0.0195 & 0.0244& 0.0508& 0.0852& 0.0260& 0.0331\\
         \textbf{LayerGCN} & 0.0529& 0.0820& 0.0281& 0.0355& 0.0594& 0.0916& 0.0323& 0.0406& 0.0371& 0.0566& 0.0200& 0.0247& 0.0522& 0.0866& 0.0270& 0.0346\\
         \textbf{NLGCL} & 0.0588& 0.0902& 0.0307& 0.0396& 0.0650& 0.0971& 0.0342& 0.0432& 0.0430& 0.0628& 0.0235& 0.0289& 0.0568& 0.0923& 0.0305& 0.0387\\
         \midrule
         \textbf{MMGCN} & 0.0378 &0.0615& 0.0200 &0.0261& 0.0370 &0.0605& 0.0193 &0.0254 &0.0218& 0.0345& 0.0110& 0.0142& 0.0485& 0.0862& 0.0200& 0.0298\\
         \textbf{DualGNN} & 0.0448 &0.0716& 0.0240 &0.0309& 0.0568 &0.0859& 0.0310 &0.0385 &0.0454 &0.0683 &0.0241 &0.0299& 0.0529& 0.0892& 0.0234& 0.0343\\
         \textbf{LATTICE} & 0.0547 &0.0850& 0.0292 &0.0370& 0.0620 &0.0953& 0.0335 &0.0421 &0.0492 &0.0733 &0.0268 &0.0330& 0.0559& 0.0904& 0.0303& 0.0388\\
         \textbf{SLMRec} & 0.0529 &0.0775& 0.0290 &0.0353& 0.0663 &0.0990& 0.0365 &0.0450 &0.0452 &0.0675 &0.0247 &0.0303& 0.0533& 0.0870& 0.0258& 0.0340\\
         \textbf{BM3} & 0.0564 &0.0883& 0.0301 &0.0383& 0.0656 &0.0980& 0.0355 &0.0438 &0.0422 &0.0621 &0.0231 &0.0281& 0.0510& 0.0882& 0.0221& 0.0329\\
         \textbf{MMSSL} & 0.0613 &0.0971& 0.0326 &0.0420& 0.0693 &0.1013& 0.0369 &0.0474 &0.0531 &0.0797 &0.0291 &0.0359& 0.0540& 0.0891& 0.0271& 0.0355\\
         \textbf{FREEDOM} & 0.0627 &0.0992& 0.0330 & 0.0424& 0.0717&0.1089& 0.0385 &0.0481 &0.0629 &0.0941 &0.0341 &0.0420& 0.0575& 0.0900& 0.0316 & 0.0398\\
         \textbf{LGMRec} & 0.0639& 0.0989& 0.0337& 0.0430& 0.0719& 0.1068& 0.0387& 0.0477& 0.0555& 0.0828& 0.0302& 0.0371 & 0.0588& 0.0923& 0.0312& 0.0400\\
         \textbf{DiffMM} & 0.0623& 0.0975& 0.0328& 0.0411& 0.0671& 0.1017& 0.0377& 0.0458& 0.0531& 0.0797& 0.0291& 0.0359& 0.0551 & 0.0868& 0.0301& 0.0377\\
         \textbf{SMORE} & \underline{0.0680}& 0.1035& \underline{0.0365}& \underline{0.0457}& 0.0762& \underline{0.1142}& 0.0408& 0.0506& 0.0659& 0.0987& \underline{0.0360}& 0.0443& 0.0623& 0.0992& 0.0328& 0.0417\\
         \textbf{HPMRec} & 0.0667& 0.1033& 0.0357& 0.0451& 0.0751& 0.1129& \underline{0.0410}& 0.0507& 0.0658& 0.0963& 0.0351& 0.0429& 0.0620& 0.0985& 0.0320& 0.0408\\
         \textbf{MENTOR} & 0.0678& 0.1048& 0.0362& 0.0450& \underline{0.0763}& 0.1139& 0.0409& \underline{0.0511}& \underline{0.0668}& \underline{0.0989}& \underline{0.0360}& 0.0441& \underline{0.0634}& \underline{0.1011}& 0.0334& 0.0426\\
         \textbf{COHESION} & \underline{0.0680}& \underline{0.1052}& 0.0354& 0.0454& 0.0752& 0.1137& 0.0409& 0.0503& 0.0665& 0.0983& 0.0358& 0.0438& 0.0630& 0.1008& \underline{0.0337}& \underline{0.0431}\\
          \midrule
         \textbf{IIMRec} & \textbf{0.0709$^*$}& \textbf{0.1094$^*$}& \textbf{0.0387$^*$}& \textbf{0.0485$^*$}& \textbf{0.0793$^*$}& \textbf{0.1184$^*$}& \textbf{0.0431$^*$}& \textbf{0.0537$^*$}& \textbf{0.0702$^*$}& \textbf{0.1032$^*$}& \textbf{0.0379$^*$}& \textbf{0.0469$^*$}& \textbf{0.0669$^*$}& \textbf{0.1053$^*$}& \textbf{0.0353$^*$}& \textbf{0.0455$^*$}\\
         \bottomrule
    \end{tabular}
    }
    \vskip -0.1in
\end{table*}
\section{Experiment}
\label{sec:Experiment}

\subsection{Settings}
\subsubsection{Evaluation Datasets}
We conduct experiments on four real-world datasets spanning multiple modalities. The Baby, Sports, and Clothing subsets are derived from the Amazon dataset \cite{mcauley2015image}, where each item is accompanied by both textual descriptions and images. To ensure a fair comparison, we adopt the data preprocessing strategy from MMRec \cite{zhou2023mmrecsm} for these three datasets. Additionally, we incorporate the TikTok dataset \cite{jiang2024diffmm}, which consists of user interaction logs with short videos and encompasses more than two modalities, enabling us to evaluate IIMRec under richer multimodal conditions. The preprocessing of TikTok follows the protocols established in prior works \cite{jiang2024diffmm, wei2023multi}. Dataset statistics are presented in Table~\ref{tab:dataset_statistics}.

\subsubsection{Evaluation Baselines}
To evaluate the effectiveness of IIMRec, we compare it with advanced models. Based on whether multimodal features are used, these methods fall into two categories: general recommendation (\textbf{LightGCN \cite{he2020lightgcn}, SimGCL \cite{yu2022graph}, LayerGCN \cite{zhou2023layer}, and NLGCL \cite{xu2025nlgcl}}) and multimodal recommendation (\textbf{MMGCN \citep{wei2019mmgcn}, DualGNN \citep{wang2021dualgnn}, LATTICE \citep{zhang2021mining}, SLMRec \citep{tao2022self}, BM3 \citep{zhou2023bootstrap}, MMSSL \citep{wei2023multi}, FREEDOM \citep{zhou2023tale}, LGMRec \citep{guo2024lgmrec}, DiffMM \citep{jiang2024diffmm}, SMORE \citep{ong2025spectrum}, HPMRec \citep{chen2025hypercomplex}, MENTOR \citep{xu2024mentor}, and COHESION \citep{xu2025cohesion}}). Detailed descriptions of these baselines are provided in Appendix~\ref{appendix:baselines}.

\subsubsection{Evaluation Protocols}
To ensure a fair comparison, we adopt the evaluation protocol established in prior studies \citep{xu2025survey,zhou2023bootstrap,zhou2023tale}. Specifically, the optimal model is chosen according to the highest Recall@20 achieved on the validation set. The average performance across all users in the test set is evaluated using Recall@10, Recall@20, NDCG@10, and NDCG@20.

\subsubsection{Implementation Details}
To be fair, we set both the user and item embedding dimensions to 64 and the training batch size to 2,048 to optimize all models. We search for optimal parameter ranges in the validation set as reported by their original baseline models for all models, and report the corresponding test set results. For our \method{}, we conduct a grid search over the hyper-parameters, specifically $k \in \{5, 10, 15, 20\}$ for all $\operatorname{Top\text{-}k}(\cdot)$ operations, $\kappa \in \{0.2, 0.4, 0.6, 0.8\}$, $\gamma \in \{0.2, 0.4, 0.6, 0.8\}$, $\delta \in \{0.2, 0.4, 0.6, 0.8\}$, $\lambda_{\text{INA}} \in \{0.2, 0.4, 0.6, 0.8\}$, and $\lambda_{\text{CL}} \in \{1\mathrm{e}^{-5}, 1\mathrm{e}^{-4}, 1\mathrm{e}^{-3}\}$. We fix $\tau = 0.2$ for all softmax operations to reduce the hyper-parameter tuning cost. An NVIDIA RTX 4090 GPU was used for all experiments.

\subsection{Overall Performance Comparison}
We evaluate the effectiveness of IIMRec on multiple real-world datasets in multimodal recommendation scenarios. The optimal results are highlighted in bold, while the suboptimal ones are underlined. From Table~\ref{tab:result}, we derive the following observations:

\begin{itemize}[leftmargin=*]
    \item IIMRec achieves the best performance across all four datasets on every evaluation metric, with statistically significant improvements over the strongest baselines. The gains hold across datasets of different domains, interaction densities, and modality configurations, confirming that our high-quality item-item graph generalizes beyond any existing graph structure. In particular, the improvements are most evident under sparse-interaction conditions, where the item-item graph is most susceptible to noisy edges and where fixed aggregation lacks the flexibility to suppress corrupted propagation signals.

    \item Among multimodal baselines, models incorporating item-item homogeneous graphs (SMORE, HPMRec, MENTOR, COHESION) consistently outperform those relying solely on user-item bipartite propagation (MMGCN, DualGNN, BM3), confirming the value of item-level semantic affinity as a complementary channel. However, the gap between these models and IIMRec reveals that graph construction quality and signal integration strategy matter as much as the presence of such graphs. Notably, general baselines with advanced contrastive learning sometimes rival earlier multimodal models, suggesting that naively injecting modality features through noisy item-item graphs can offset the benefit of multimodal information.
\end{itemize}

\begin{figure}[!t]
\centering
\includegraphics[width=\linewidth]{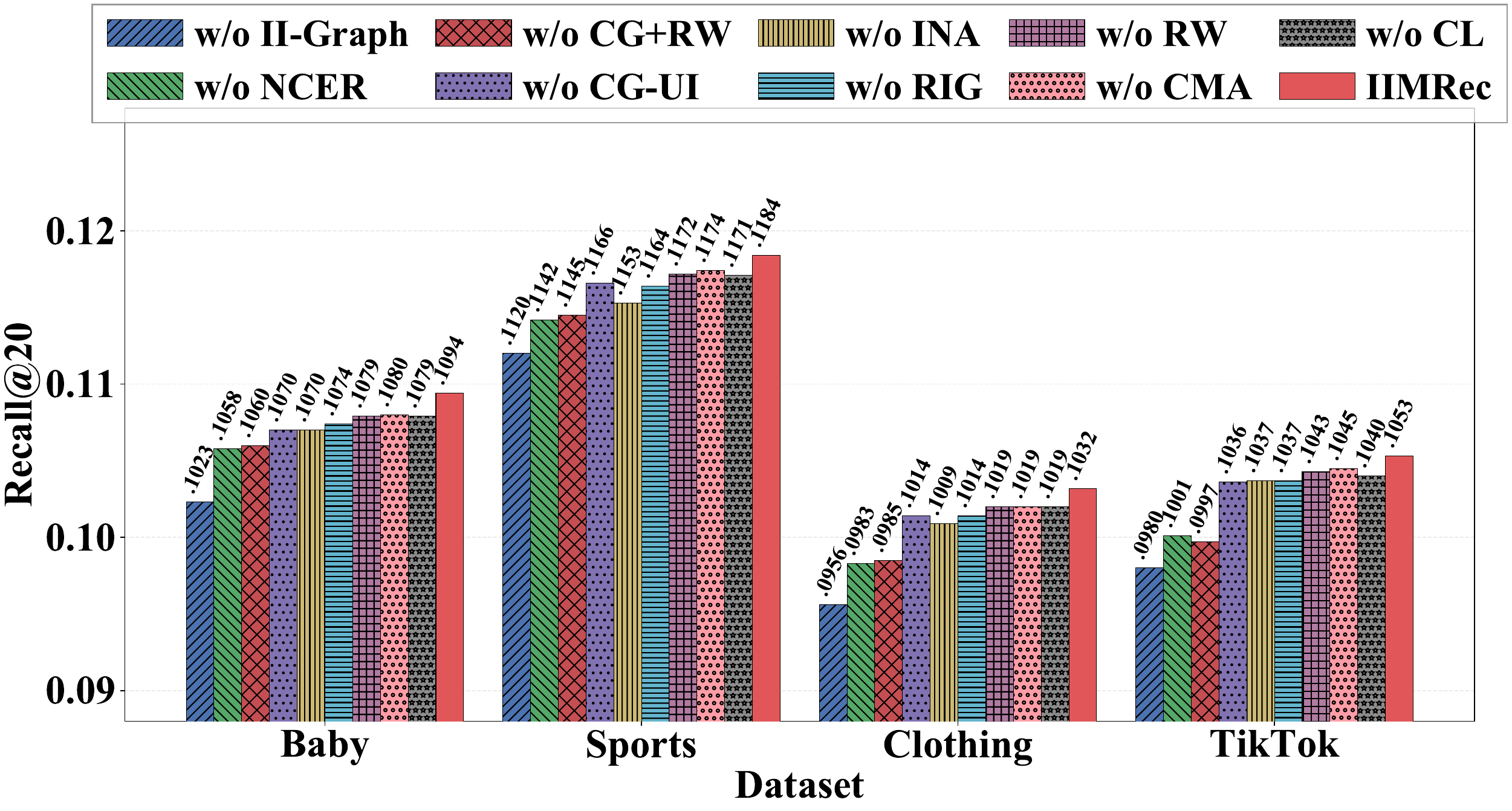}
\vskip -0.15in
\caption{Ablation study (Recall@20) on four datasets.}
\label{fig:ablation}
\vskip -0.15in
\end{figure}

\subsection{Ablation Study}
To verify the contribution of each component, we systematically remove individual modules from IIMRec and report Recall@20 on all four datasets. The variants include: removing the entire item-item graph (\textit{w/o II-Graph}), NCER reweighting (\textit{w/o NCER}), the content-guided UI graph (\textit{w/o CG-UI}), residual-weighted propagation (\textit{w/o RW}), both CG-UI and RW jointly (\textit{w/o CG+RW}), cross-modal attention fusion (\textit{w/o CMA}), INA augmentation (\textit{w/o INA}), the residual II gate (\textit{w/o RIG}), and contrastive learning (\textit{w/o CL}). 
 
As shown in Figure~\ref{fig:ablation}, removing the item-item graph entirely leads to the largest degradation across all datasets, underscoring that item-item propagation provides essential semantic information beyond user-item collaborative signals. NCER contributes the second-largest individual effect, indicating that the raw fused graph carries substantial noise that triadic-closure-based reweighting effectively suppresses. Among all core innovations, RIG and INA each produce consistent drops when removed, validating that item-adaptive gating and soft-positive augmentation address complementary weaknesses in the standard pipeline. Notably, removing CG-UI and RW simultaneously leads to a larger drop than either ablation alone, suggesting a synergistic interaction: the content-guided graph introduces virtual interactions while residual-weighted propagation prevents deeper layers from amplifying noise through these expanded connections. Cross-modal attention and contrastive learning show smaller but stable contributions as supplementary alignment mechanisms.

\begin{figure}[!t] 
    \centering
    \includegraphics[width=1\linewidth]{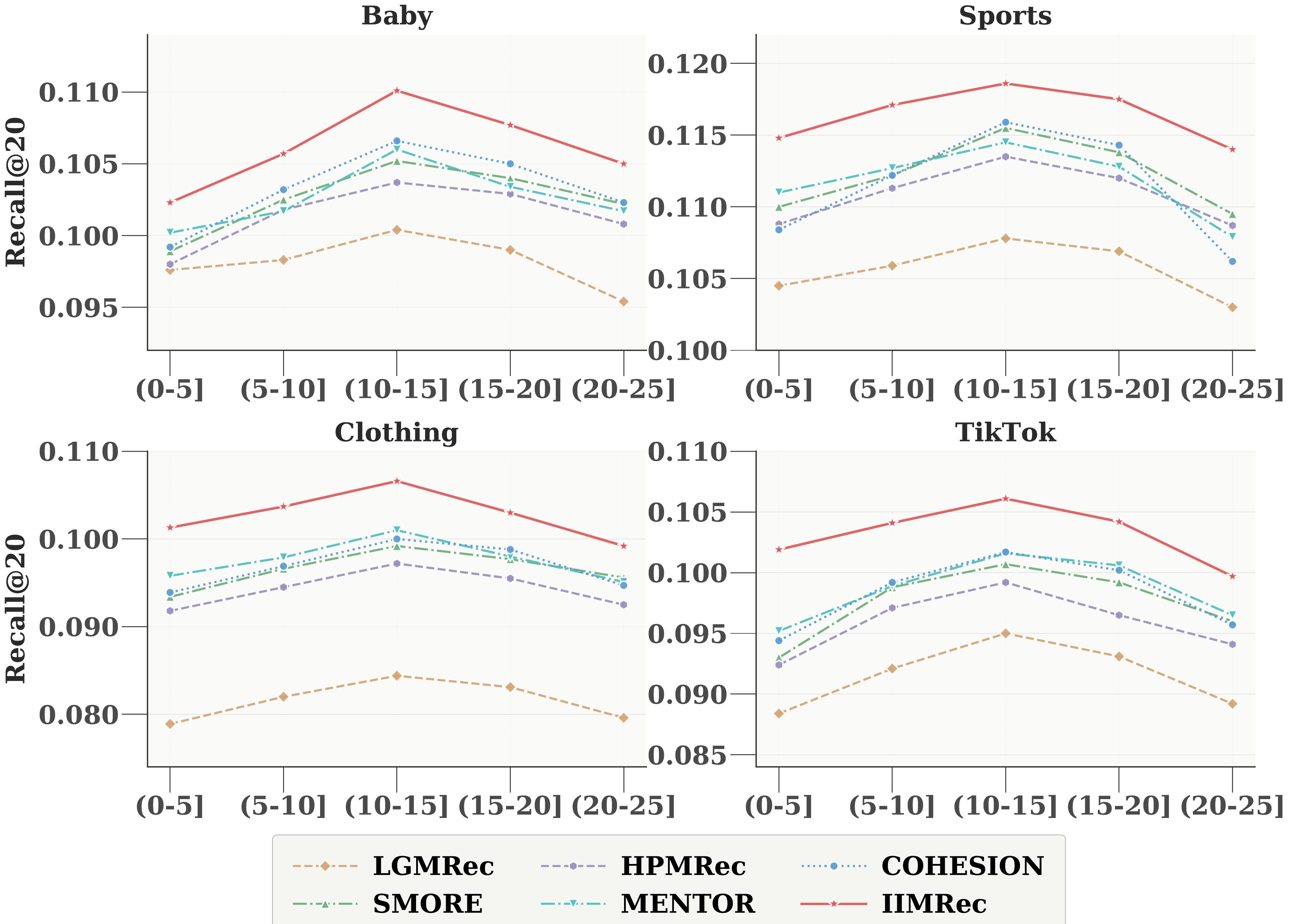}
    \vskip -0.15in
    \caption{Sparsity analysis for IIMRec across all datasets.}
    \label{fig:sparsity}
    \vskip -0.25in
\end{figure}

\begin{table}[!t]
    \centering
    \caption{Cold-start analysis across all datasets.}
    \label{tab:cold-start}
    \vskip -0.15in
    \resizebox{\linewidth}{!}{
    \begin{tabular}{ccccccccc}
    \toprule
    \multirow{2}{*}{\textbf{Method}} & \multicolumn{2}{c}{\textbf{Baby}} & \multicolumn{2}{c}{\textbf{Sports}} & \multicolumn{2}{c}{\textbf{Clothing}} & \multicolumn{2}{c}{\textbf{TikTok}} \\
    \cmidrule(lr){2-3} \cmidrule(lr){4-5} \cmidrule(lr){6-7} \cmidrule(lr){8-9}
    & R@20 & N@20 & R@20 & N@20 & R@20 & N@20 & R@20 & N@20 \\
    \midrule
    MMGCN    & 0.0186 & 0.0098 & 0.0178 & 0.0094 & 0.0101 & 0.0050 & 0.0178 & 0.0104 \\
    DualGNN  & 0.0200 & 0.0110 & 0.0242 & 0.0132 & 0.0198 & 0.0113 & 0.0224 & 0.0121 \\
    LATTICE  & 0.0256 & 0.0138 & 0.0340 & 0.0189 & 0.0249 & 0.0135 & 0.0233 & 0.0127 \\
    SLMRec   & 0.0259 & 0.0140 & 0.0354 & 0.0194 & 0.0208 & 0.0118 & 0.0243 & 0.0131 \\
    BM3      & 0.0262 & 0.0133 & 0.0294 & 0.0180 & 0.0189 & 0.0104 & 0.0249 & 0.0139 \\
    MMSSL    & 0.0351 & 0.0192 & 0.0370 & 0.0203 & 0.0294 & 0.0157 & 0.0345 & 0.0206 \\
    FREEDOM  & 0.0588 & 0.0257 & 0.0640 & 0.0289 & 0.0585 & 0.0252 & 0.0570 & 0.0239 \\
    LGMRec   & 0.0592 & 0.0261 & 0.0629 & 0.0281 & 0.0551 & 0.0235 & 0.0578 & 0.0243 \\
    DiffMM   & 0.0552 & 0.0238 & 0.0589 & 0.0254 & 0.0510 & 0.0221 & 0.0519 & 0.0230 \\
    SMORE    & 0.0595 & 0.0251 & 0.0661 & \underline{0.0302} & 0.0602 & 0.0259 & 0.0613 & 0.0255 \\
    HPMRec   & 0.0603 & 0.0256 & 0.0653 & 0.0292 & 0.0600 & 0.0255 & 0.0609 & 0.0252 \\
    MENTOR   & 0.0628 & \underline{0.0268} & 0.0661 & 0.0297 & 0.0610 & \underline{0.0264} & 0.0618 & 0.0262 \\
    COHESION & \underline{0.0631} & 0.0263 & \underline{0.0665} & 0.0300 & \underline{0.0611} & 0.0256 & \underline{0.0622} & \underline{0.0264} \\
    \midrule
    \textbf{IIMRec} & \textbf{0.0663} & \textbf{0.0279} & \textbf{0.0711} & \textbf{0.0318} & \textbf{0.0638} & \textbf{0.0270} & \textbf{0.0664} & \textbf{0.0283} \\
    \bottomrule
    \end{tabular}
    }
    \vskip -0.15in
\end{table}

\subsection{Sparsity and Cold-Start Analysis}
For sparsity analysis, we group users by their number of training interactions and compare IIMRec against five advanced baselines. As shown in Figure~\ref{fig:sparsity}, IIMRec consistently leads across all datasets and sparsity levels, with the largest margins in the sparsest groups $(0$-$5]$ and $(5$-$10]$. This confirms that NCER effectively filters unreliable edges when interaction evidence is scarce, while RIG prevents noisy item-item propagation from overwhelming the limited collaborative signal of low-activity users. For cold-start analysis, we follow prior setting~\cite{xu2025enhancing}: 20\% of items are removed from training, split evenly into validation and test sets. As shown in Table~\ref{tab:cold-start}, IIMRec achieves the best performance across all datasets. The gains stem from the synergy of all components: NCER provides high-quality semantic neighbors for cold-start items, RIG controls their contribution to avoid noise from uncertain modality features, and INA propagates supervision to items with few or no observed interactions via discounted soft positives.

\begin{table}
    \centering
    \caption{Efficiency analysis across all datasets. Time: seconds per epoch; Mem.: GPU memory in GB.}
    \label{tab:efficiency}
    \small
    \vskip -0.15in
\resizebox{\linewidth}{!}{
    \begin{tabular}{l*{4}{rr}}
        \toprule
        & \multicolumn{2}{c}{Baby} & \multicolumn{2}{c}{Sports} & \multicolumn{2}{c}{Clothing} & \multicolumn{2}{c}{TikTok} \\
        \cmidrule(lr){2-3} \cmidrule(lr){4-5} \cmidrule(lr){6-7} \cmidrule(lr){8-9}
        Method & Time & Mem. & Time & Mem. & Time & Mem. & Time & Mem. \\
        \midrule
        LGMRec   & 5.93  & 2.41  & 8.98  & 3.67  & 10.02 & 4.81  & 1.90 & 1.08 \\
        SMORE    & 6.55  & 3.31  & 9.29  & 5.02  & 11.05 & 6.89  & 2.09 & 1.39 \\
        HPMRec   & 21.03 & 8.58  & 30.85 & 10.19 & 40.23 & 14.95 & 6.69 & 2.60 \\
        MENTOR   & 7.03  & 7.12  & 9.62  & 8.44  & 11.90 & 12.99 & 2.18 & 2.33 \\
        COHESION & 4.47  & 2.89  & 7.91  & 4.20  & 9.05  & 5.73  & 1.59 & 1.31 \\
        \midrule
        IIMRec  & \textbf{3.18} & \textbf{1.08} & \textbf{6.02} & \textbf{2.12} & \textbf{6.89} & \textbf{2.69} & \textbf{0.75} & \textbf{0.77} \\
        \bottomrule
    \end{tabular}
    }
    \vskip -0.15in
\end{table}

\subsection{Efficiency Study}
\label{sec:efficiency}
Table~\ref{tab:efficiency} reports the per-epoch training time and GPU memory usage of IIMRec and five competitive baselines. Although IIMRec involves a comprehensive graph construction, the high-quality item-item graph is precomputed from multimodal raw features and historical interactions before training and cached to disk, introducing zero overhead during training. As a result, IIMRec not only achieves consistent performance gains over LGMRec, SMORE, HPMRec, MENTOR, and COHESION, but also runs faster and consumes less GPU memory than every compared method across all four datasets, making it well-suited for real-world deployment.

\subsection{Hyper-parameter Analysis}
\label{sec:hyper}
We investigate the sensitivity of IIMRec to six key hyper-parameters by varying each while fixing the others at their optimal values, and report Recall@20 across all four datasets in Figure~\ref{fig:hyperparameter} in Appendix~\ref{appendix:hyper}.

\section{Conclusion}
\label{sec:conclusion}
This work presents \method{}, a framework that constructs a single high-quality item-item graph and systematically reuses it across three stages of the multimodal recommendation. NCER refines the graph by reweighting edges according to neighborhood overlap grounded in triadic closure, amplifying structurally reliable connections while suppressing spurious ones. RIG introduces a learned per-item gate that adaptively controls how much item-item propagation signal each item absorbs, replacing the uniform fixed aggregation universally adopted in prior work. INA enriches BPR training by treating top neighbors of positive items as discounted soft positives, providing additional supervision that pushes in the same direction as preference learning. Together with content-guided UI graph expansion, these components realize a "\textit{build once, leverage everywhere}" paradigm that is lightweight, modular, and theoretically grounded. Experiments on four dataset demonstrate consistent state-of-the-art performance with lower computational cost, and particularly strong improvements under cold-start and sparse-interaction conditions.

\begin{acks}
This work was supported by the UGC General Research Fund no. 17209822 and the Innovation and Technology Commission Fund no. ITS/383/23FP from Hong Kong.
\end{acks}

\balance
\bibliographystyle{ACM-Reference-Format}

\begin{thebibliography}{37}


\ifx \showCODEN    \undefined \def \showCODEN     #1{\unskip}     \fi
\ifx \showISBNx    \undefined \def \showISBNx     #1{\unskip}     \fi
\ifx \showISBNxiii \undefined \def \showISBNxiii  #1{\unskip}     \fi
\ifx \showISSN     \undefined \def \showISSN      #1{\unskip}     \fi
\ifx \showLCCN     \undefined \def \showLCCN      #1{\unskip}     \fi
\ifx \shownote     \undefined \def \shownote      #1{#1}          \fi
\ifx \showarticletitle \undefined \def \showarticletitle #1{#1}   \fi
\ifx \showURL      \undefined \def \showURL       {\relax}        \fi
\providecommand\bibfield[2]{#2}
\providecommand\bibinfo[2]{#2}
\providecommand\natexlab[1]{#1}
\providecommand\showeprint[2][]{arXiv:#2}

\bibitem[Chen et~al\mbox{.}(2025a)]%
        {chen2025don}
\bibfield{author}{\bibinfo{person}{Zheyu Chen}, \bibinfo{person}{Jinfeng Xu}, {and} \bibinfo{person}{Haibo Hu}.} \bibinfo{year}{2025}\natexlab{a}.
\newblock \showarticletitle{Don’t Lose Yourself: Boosting Multimodal Recommendation via Reducing Node-neighbor Discrepancy in Graph Convolutional Network}. In \bibinfo{booktitle}{\emph{ICASSP 2025-2025 IEEE International Conference on Acoustics, Speech and Signal Processing (ICASSP)}}. IEEE, \bibinfo{pages}{1--5}.
\newblock


\bibitem[Chen et~al\mbox{.}(2025b)]%
        {chen2025hypercomplex}
\bibfield{author}{\bibinfo{person}{Zheyu Chen}, \bibinfo{person}{Jinfeng Xu}, \bibinfo{person}{Hewei Wang}, \bibinfo{person}{Shuo Yang}, \bibinfo{person}{Zitong Wan}, {and} \bibinfo{person}{Haibo Hu}.} \bibinfo{year}{2025}\natexlab{b}.
\newblock \showarticletitle{Hypercomplex Prompt-aware Multimodal Recommendation}. In \bibinfo{booktitle}{\emph{Proceedings of the 34th ACM International Conference on Information and Knowledge Management}}. \bibinfo{pages}{403--414}.
\newblock


\bibitem[Easley et~al\mbox{.}(2010)]%
        {easley2010networks}
\bibfield{author}{\bibinfo{person}{David Easley}, \bibinfo{person}{Jon Kleinberg}, {et~al\mbox{.}}} \bibinfo{year}{2010}\natexlab{}.
\newblock \bibinfo{booktitle}{\emph{Networks, crowds, and markets: Reasoning about a highly connected world}}. Vol.~\bibinfo{volume}{1}.
\newblock \bibinfo{publisher}{Cambridge university press Cambridge}.
\newblock


\bibitem[Guo et~al\mbox{.}(2024)]%
        {guo2024lgmrec}
\bibfield{author}{\bibinfo{person}{Zhiqiang Guo}, \bibinfo{person}{Jianjun Li}, \bibinfo{person}{Guohui Li}, \bibinfo{person}{Chaoyang Wang}, \bibinfo{person}{Si Shi}, {and} \bibinfo{person}{Bin Ruan}.} \bibinfo{year}{2024}\natexlab{}.
\newblock \showarticletitle{LGMRec: Local and Global Graph Learning for Multimodal Recommendation}. In \bibinfo{booktitle}{\emph{Proceedings of the AAAI Conference on Artificial Intelligence}}, Vol.~\bibinfo{volume}{38}. \bibinfo{pages}{8454--8462}.
\newblock


\bibitem[He and McAuley(2016)]%
        {he2016vbpr}
\bibfield{author}{\bibinfo{person}{Ruining He} {and} \bibinfo{person}{Julian McAuley}.} \bibinfo{year}{2016}\natexlab{}.
\newblock \showarticletitle{VBPR: visual bayesian personalized ranking from implicit feedback}. In \bibinfo{booktitle}{\emph{Proceedings of the AAAI conference on artificial intelligence}}, Vol.~\bibinfo{volume}{30}.
\newblock


\bibitem[He et~al\mbox{.}(2020)]%
        {he2020lightgcn}
\bibfield{author}{\bibinfo{person}{Xiangnan He}, \bibinfo{person}{Kuan Deng}, \bibinfo{person}{Xiang Wang}, \bibinfo{person}{Yan Li}, \bibinfo{person}{Yongdong Zhang}, {and} \bibinfo{person}{Meng Wang}.} \bibinfo{year}{2020}\natexlab{}.
\newblock \showarticletitle{Lightgcn: Simplifying and powering graph convolution network for recommendation}. In \bibinfo{booktitle}{\emph{Proceedings of the 43rd International ACM SIGIR conference on research and development in Information Retrieval}}. \bibinfo{pages}{639--648}.
\newblock


\bibitem[Jiang et~al\mbox{.}(2024)]%
        {jiang2024diffmm}
\bibfield{author}{\bibinfo{person}{Yangqin Jiang}, \bibinfo{person}{Lianghao Xia}, \bibinfo{person}{Wei Wei}, \bibinfo{person}{Da Luo}, \bibinfo{person}{Kangyi Lin}, {and} \bibinfo{person}{Chao Huang}.} \bibinfo{year}{2024}\natexlab{}.
\newblock \showarticletitle{DiffMM: Multi-Modal Diffusion Model for Recommendation}.
\newblock  (\bibinfo{year}{2024}).
\newblock


\bibitem[Kim et~al\mbox{.}(2024)]%
        {kim2024monet}
\bibfield{author}{\bibinfo{person}{Yungi Kim}, \bibinfo{person}{Taeri Kim}, \bibinfo{person}{Won-Yong Shin}, {and} \bibinfo{person}{Sang-Wook Kim}.} \bibinfo{year}{2024}\natexlab{}.
\newblock \showarticletitle{MONET: Modality-Embracing Graph Convolutional Network and Target-Aware Attention for Multimedia Recommendation}. In \bibinfo{booktitle}{\emph{Proceedings of the 17th ACM International Conference on Web Search and Data Mining}}. \bibinfo{pages}{332--340}.
\newblock


\bibitem[Li et~al\mbox{.}(2025)]%
        {li2025ddunet}
\bibfield{author}{\bibinfo{person}{Yijie Li}, \bibinfo{person}{Hewei Wang}, \bibinfo{person}{Jinfeng Xu}, \bibinfo{person}{Puzhen Wu}, \bibinfo{person}{Yunzhong Xiao}, \bibinfo{person}{Shaofan Wang}, {and} \bibinfo{person}{Soumyabrata Dev}.} \bibinfo{year}{2025}\natexlab{}.
\newblock \showarticletitle{Ddunet: Dual dynamic u-net for highly-efficient cloud segmentation}. In \bibinfo{booktitle}{\emph{IGARSS 2025-2025 IEEE International Geoscience and Remote Sensing Symposium}}. IEEE, \bibinfo{pages}{4705--4711}.
\newblock


\bibitem[McAuley et~al\mbox{.}(2015)]%
        {mcauley2015image}
\bibfield{author}{\bibinfo{person}{Julian McAuley}, \bibinfo{person}{Christopher Targett}, \bibinfo{person}{Qinfeng Shi}, {and} \bibinfo{person}{Anton Van Den~Hengel}.} \bibinfo{year}{2015}\natexlab{}.
\newblock \showarticletitle{Image-based recommendations on styles and substitutes}. In \bibinfo{booktitle}{\emph{Proceedings of the 38th international ACM SIGIR conference on research and development in information retrieval}}. \bibinfo{pages}{43--52}.
\newblock


\bibitem[Natarajan et~al\mbox{.}(2013)]%
        {natarajan2013learning}
\bibfield{author}{\bibinfo{person}{Nagarajan Natarajan}, \bibinfo{person}{Inderjit~S Dhillon}, \bibinfo{person}{Pradeep~K Ravikumar}, {and} \bibinfo{person}{Ambuj Tewari}.} \bibinfo{year}{2013}\natexlab{}.
\newblock \showarticletitle{Learning with noisy labels}.
\newblock \bibinfo{journal}{\emph{Advances in neural information processing systems}}  \bibinfo{volume}{26} (\bibinfo{year}{2013}).
\newblock


\bibitem[Ong and Khong(2025)]%
        {ong2025spectrum}
\bibfield{author}{\bibinfo{person}{Rongqing~Kenneth Ong} {and} \bibinfo{person}{Andy~WH Khong}.} \bibinfo{year}{2025}\natexlab{}.
\newblock \showarticletitle{Spectrum-based modality representation fusion graph convolutional network for multimodal recommendation}. In \bibinfo{booktitle}{\emph{Proceedings of the Eighteenth ACM International Conference on Web Search and Data Mining}}. \bibinfo{pages}{773--781}.
\newblock


\bibitem[Oord et~al\mbox{.}(2018)]%
        {oord2018representation}
\bibfield{author}{\bibinfo{person}{Aaron van~den Oord}, \bibinfo{person}{Yazhe Li}, {and} \bibinfo{person}{Oriol Vinyals}.} \bibinfo{year}{2018}\natexlab{}.
\newblock \showarticletitle{Representation learning with contrastive predictive coding}.
\newblock \bibinfo{journal}{\emph{arXiv preprint arXiv:1807.03748}} (\bibinfo{year}{2018}).
\newblock


\bibitem[Rendle et~al\mbox{.}(2012)]%
        {rendle2012bpr}
\bibfield{author}{\bibinfo{person}{Steffen Rendle}, \bibinfo{person}{Christoph Freudenthaler}, \bibinfo{person}{Zeno Gantner}, {and} \bibinfo{person}{Lars Schmidt-Thieme}.} \bibinfo{year}{2012}\natexlab{}.
\newblock \showarticletitle{BPR: Bayesian personalized ranking from implicit feedback}.
\newblock \bibinfo{journal}{\emph{arXiv preprint arXiv:1205.2618}} (\bibinfo{year}{2012}).
\newblock


\bibitem[Tang et~al\mbox{.}(2019)]%
        {tang2019adversarial}
\bibfield{author}{\bibinfo{person}{Jinhui Tang}, \bibinfo{person}{Xiaoyu Du}, \bibinfo{person}{Xiangnan He}, \bibinfo{person}{Fajie Yuan}, \bibinfo{person}{Qi Tian}, {and} \bibinfo{person}{Tat-Seng Chua}.} \bibinfo{year}{2019}\natexlab{}.
\newblock \showarticletitle{Adversarial training towards robust multimedia recommender system}.
\newblock \bibinfo{journal}{\emph{IEEE Transactions on Knowledge and Data Engineering}} \bibinfo{volume}{32}, \bibinfo{number}{5} (\bibinfo{year}{2019}), \bibinfo{pages}{855--867}.
\newblock


\bibitem[Tao et~al\mbox{.}(2022)]%
        {tao2022self}
\bibfield{author}{\bibinfo{person}{Zhulin Tao}, \bibinfo{person}{Xiaohao Liu}, \bibinfo{person}{Yewei Xia}, \bibinfo{person}{Xiang Wang}, \bibinfo{person}{Lifang Yang}, \bibinfo{person}{Xianglin Huang}, {and} \bibinfo{person}{Tat-Seng Chua}.} \bibinfo{year}{2022}\natexlab{}.
\newblock \showarticletitle{Self-supervised learning for multimedia recommendation}.
\newblock \bibinfo{journal}{\emph{IEEE Transactions on Multimedia}} (\bibinfo{year}{2022}).
\newblock


\bibitem[Wang et~al\mbox{.}(2021)]%
        {wang2021dualgnn}
\bibfield{author}{\bibinfo{person}{Qifan Wang}, \bibinfo{person}{Yinwei Wei}, \bibinfo{person}{Jianhua Yin}, \bibinfo{person}{Jianlong Wu}, \bibinfo{person}{Xuemeng Song}, {and} \bibinfo{person}{Liqiang Nie}.} \bibinfo{year}{2021}\natexlab{}.
\newblock \showarticletitle{Dualgnn: Dual graph neural network for multimedia recommendation}.
\newblock \bibinfo{journal}{\emph{IEEE Transactions on Multimedia}} (\bibinfo{year}{2021}).
\newblock


\bibitem[Wei et~al\mbox{.}(2023)]%
        {wei2023multi}
\bibfield{author}{\bibinfo{person}{Wei Wei}, \bibinfo{person}{Chao Huang}, \bibinfo{person}{Lianghao Xia}, {and} \bibinfo{person}{Chuxu Zhang}.} \bibinfo{year}{2023}\natexlab{}.
\newblock \showarticletitle{Multi-Modal Self-Supervised Learning for Recommendation}. In \bibinfo{booktitle}{\emph{Proceedings of the ACM Web Conference 2023}}. \bibinfo{pages}{790--800}.
\newblock


\bibitem[Wei et~al\mbox{.}(2020)]%
        {wei2020graph}
\bibfield{author}{\bibinfo{person}{Yinwei Wei}, \bibinfo{person}{Xiang Wang}, \bibinfo{person}{Liqiang Nie}, \bibinfo{person}{Xiangnan He}, {and} \bibinfo{person}{Tat-Seng Chua}.} \bibinfo{year}{2020}\natexlab{}.
\newblock \showarticletitle{Graph-refined convolutional network for multimedia recommendation with implicit feedback}. In \bibinfo{booktitle}{\emph{Proceedings of the 28th ACM international conference on multimedia}}. \bibinfo{pages}{3541--3549}.
\newblock


\bibitem[Wei et~al\mbox{.}(2019)]%
        {wei2019mmgcn}
\bibfield{author}{\bibinfo{person}{Yinwei Wei}, \bibinfo{person}{Xiang Wang}, \bibinfo{person}{Liqiang Nie}, \bibinfo{person}{Xiangnan He}, \bibinfo{person}{Richang Hong}, {and} \bibinfo{person}{Tat-Seng Chua}.} \bibinfo{year}{2019}\natexlab{}.
\newblock \showarticletitle{MMGCN: Multi-modal graph convolution network for personalized recommendation of micro-video}. In \bibinfo{booktitle}{\emph{Proceedings of the 27th ACM international conference on multimedia}}. \bibinfo{pages}{1437--1445}.
\newblock


\bibitem[Xu et~al\mbox{.}(2025b)]%
        {xu2025mdvt}
\bibfield{author}{\bibinfo{person}{Jinfeng Xu}, \bibinfo{person}{Zheyu Chen}, \bibinfo{person}{Jinze Li}, \bibinfo{person}{Shuo Yang}, \bibinfo{person}{Hewei Wang}, \bibinfo{person}{Yijie Li}, \bibinfo{person}{Mengran Li}, \bibinfo{person}{Puzhen Wu}, {and} \bibinfo{person}{Edith~CH Ngai}.} \bibinfo{year}{2025}\natexlab{b}.
\newblock \showarticletitle{Mdvt: Enhancing multimodal recommendation with model-agnostic multimodal-driven virtual triplets}. In \bibinfo{booktitle}{\emph{Proceedings of the 31st ACM SIGKDD Conference on Knowledge Discovery and Data Mining V. 2}}. \bibinfo{pages}{3378--3389}.
\newblock


\bibitem[Xu et~al\mbox{.}(2025a)]%
        {xu2025enhancing}
\bibfield{author}{\bibinfo{person}{Jinfeng Xu}, \bibinfo{person}{Zheyu Chen}, \bibinfo{person}{Jinze Li}, \bibinfo{person}{Shuo Yang}, \bibinfo{person}{Wei Wang}, \bibinfo{person}{Xiping Hu}, \bibinfo{person}{Raymond Chi-Wing Wong}, {and} \bibinfo{person}{Edith~CH Ngai}.} \bibinfo{year}{2025}\natexlab{a}.
\newblock \showarticletitle{Enhancing Robustness and Generalization Capability for Multimodal Recommender Systems via Sharpness-Aware Minimization}.
\newblock \bibinfo{journal}{\emph{IEEE Transactions on Knowledge and Data Engineering}} (\bibinfo{year}{2025}).
\newblock


\bibitem[Xu et~al\mbox{.}(2025c)]%
        {xu2025cohesion}
\bibfield{author}{\bibinfo{person}{Jinfeng Xu}, \bibinfo{person}{Zheyu Chen}, \bibinfo{person}{Wei Wang}, \bibinfo{person}{Xiping Hu}, \bibinfo{person}{Sang-Wook Kim}, {and} \bibinfo{person}{Edith~CH Ngai}.} \bibinfo{year}{2025}\natexlab{c}.
\newblock \showarticletitle{COHESION: Composite Graph Convolutional Network with Dual-Stage Fusion for Multimodal Recommendation}. In \bibinfo{booktitle}{\emph{Proceedings of the 48th International ACM SIGIR Conference on Research and Development in Information Retrieval}}. \bibinfo{pages}{1830--1839}.
\newblock


\bibitem[Xu et~al\mbox{.}(2025d)]%
        {xu2025best}
\bibfield{author}{\bibinfo{person}{Jinfeng Xu}, \bibinfo{person}{Zheyu Chen}, \bibinfo{person}{Shuo Yang}, \bibinfo{person}{Jinze Li}, {and} \bibinfo{person}{Edith~CH Ngai}.} \bibinfo{year}{2025}\natexlab{d}.
\newblock \showarticletitle{The Best is Yet to Come: Graph Convolution in the Testing Phase for Multimodal Recommendation}. In \bibinfo{booktitle}{\emph{Proceedings of the 33rd ACM International Conference on Multimedia}}. \bibinfo{pages}{6325--6334}.
\newblock


\bibitem[Xu et~al\mbox{.}(2025e)]%
        {xu2025vi}
\bibfield{author}{\bibinfo{person}{Jinfeng Xu}, \bibinfo{person}{Zheyu Chen}, \bibinfo{person}{Shuo Yang}, \bibinfo{person}{Jinze Li}, \bibinfo{person}{Zitong Wan}, \bibinfo{person}{Hewei Wang}, \bibinfo{person}{Weijie Liu}, \bibinfo{person}{Yijie Li}, {and} \bibinfo{person}{Edith~CH Ngai}.} \bibinfo{year}{2025}\natexlab{e}.
\newblock \showarticletitle{VI-MMRec: Similarity-Aware Training Cost-free Virtual User-Item Interactions for Multimodal Recommendation}.
\newblock \bibinfo{journal}{\emph{arXiv preprint arXiv:2512.08702}} (\bibinfo{year}{2025}).
\newblock


\bibitem[Xu et~al\mbox{.}(2025f)]%
        {xu2024mentor}
\bibfield{author}{\bibinfo{person}{Jinfeng Xu}, \bibinfo{person}{Zheyu Chen}, \bibinfo{person}{Shuo Yang}, \bibinfo{person}{Jinze Li}, \bibinfo{person}{Hewei Wang}, {and} \bibinfo{person}{Edith~CH Ngai}.} \bibinfo{year}{2025}\natexlab{f}.
\newblock \showarticletitle{Mentor: multi-level self-supervised learning for multimodal recommendation}. In \bibinfo{booktitle}{\emph{Proceedings of the AAAI Conference on Artificial Intelligence}}, Vol.~\bibinfo{volume}{39}. \bibinfo{pages}{12908--12917}.
\newblock


\bibitem[Xu et~al\mbox{.}(2026b)]%
        {xu2026well}
\bibfield{author}{\bibinfo{person}{Jinfeng Xu}, \bibinfo{person}{Zheyu Chen}, \bibinfo{person}{Shuo Yang}, \bibinfo{person}{Jinze Li}, \bibinfo{person}{Hewei Wang}, \bibinfo{person}{Jianheng Tang}, \bibinfo{person}{Wei Wang}, \bibinfo{person}{Xiping Hu}, {and} \bibinfo{person}{Edith~CH Ngai}.} \bibinfo{year}{2026}\natexlab{b}.
\newblock \showarticletitle{Well Begun is Half Done: Training-Free and Model-Agnostic Semantically Guaranteed User Representation Initialization for Multimodal Recommendation}. In \bibinfo{booktitle}{\emph{Proceedings of the 49th International ACM SIGIR Conference on Research and Development in Information Retrieval}}. \bibinfo{pages}{2096--2106}.
\newblock


\bibitem[Xu et~al\mbox{.}(2025g)]%
        {xu2025nlgcl}
\bibfield{author}{\bibinfo{person}{Jinfeng Xu}, \bibinfo{person}{Zheyu Chen}, \bibinfo{person}{Shuo Yang}, \bibinfo{person}{Jinze Li}, \bibinfo{person}{Hewei Wang}, \bibinfo{person}{Wei Wang}, \bibinfo{person}{Xiping Hu}, {and} \bibinfo{person}{Edith Ngai}.} \bibinfo{year}{2025}\natexlab{g}.
\newblock \showarticletitle{NLGCL: Naturally Existing Neighbor Layers Graph Contrastive Learning for Recommendation}. In \bibinfo{booktitle}{\emph{Proceedings of the Nineteenth ACM Conference on Recommender Systems}}. \bibinfo{pages}{319--329}.
\newblock


\bibitem[Xu et~al\mbox{.}(2026c)]%
        {xu2026nlgclp}
\bibfield{author}{\bibinfo{person}{Jinfeng Xu}, \bibinfo{person}{Zheyu Chen}, \bibinfo{person}{Shuo Yang}, \bibinfo{person}{Jinze Li}, \bibinfo{person}{Hewei Wang}, \bibinfo{person}{Wei Wang}, \bibinfo{person}{Xiping Hu}, {and} \bibinfo{person}{Edith Ngai}.} \bibinfo{year}{2026}\natexlab{c}.
\newblock \showarticletitle{NLGCL+: Naturally Existing Neighbour Layers Graph Contrastive Learning with Adaptive Sample Weighting for Multimodal Recommendation}.
\newblock \bibinfo{journal}{\emph{ACM Transactions on Recommender Systems}} (\bibinfo{year}{2026}).
\newblock


\bibitem[Xu et~al\mbox{.}(2026a)]%
        {xu2025survey}
\bibfield{author}{\bibinfo{person}{Jinfeng Xu}, \bibinfo{person}{Zheyu Chen}, \bibinfo{person}{Shuo Yang}, \bibinfo{person}{Jinze Li}, \bibinfo{person}{Wei Wang}, \bibinfo{person}{Xiping Hu}, \bibinfo{person}{Steven Hoi}, {and} \bibinfo{person}{Edith Ngai}.} \bibinfo{year}{2026}\natexlab{a}.
\newblock \showarticletitle{A survey on multimodal recommender systems: Recent advances and future directions}.
\newblock \bibinfo{journal}{\emph{IEEE Transactions on Multimedia}} (\bibinfo{year}{2026}).
\newblock


\bibitem[Yu et~al\mbox{.}(2022)]%
        {yu2022graph}
\bibfield{author}{\bibinfo{person}{Junliang Yu}, \bibinfo{person}{Hongzhi Yin}, \bibinfo{person}{Xin Xia}, \bibinfo{person}{Tong Chen}, \bibinfo{person}{Lizhen Cui}, {and} \bibinfo{person}{Quoc Viet~Hung Nguyen}.} \bibinfo{year}{2022}\natexlab{}.
\newblock \showarticletitle{Are graph augmentations necessary? simple graph contrastive learning for recommendation}. In \bibinfo{booktitle}{\emph{Proceedings of the 45th international ACM SIGIR conference on research and development in information retrieval}}. \bibinfo{pages}{1294--1303}.
\newblock


\bibitem[Zhang et~al\mbox{.}(2021)]%
        {zhang2021mining}
\bibfield{author}{\bibinfo{person}{Jinghao Zhang}, \bibinfo{person}{Yanqiao Zhu}, \bibinfo{person}{Qiang Liu}, \bibinfo{person}{Shu Wu}, \bibinfo{person}{Shuhui Wang}, {and} \bibinfo{person}{Liang Wang}.} \bibinfo{year}{2021}\natexlab{}.
\newblock \showarticletitle{Mining latent structures for multimedia recommendation}. In \bibinfo{booktitle}{\emph{Proceedings of the 29th ACM International Conference on Multimedia}}. \bibinfo{pages}{3872--3880}.
\newblock


\bibitem[Zhou et~al\mbox{.}(2023c)]%
        {zhou2023comprehensive}
\bibfield{author}{\bibinfo{person}{Hongyu Zhou}, \bibinfo{person}{Xin Zhou}, \bibinfo{person}{Zhiwei Zeng}, \bibinfo{person}{Lingzi Zhang}, {and} \bibinfo{person}{Zhiqi Shen}.} \bibinfo{year}{2023}\natexlab{c}.
\newblock \showarticletitle{A comprehensive survey on multimodal recommender systems: Taxonomy, evaluation, and future directions}.
\newblock \bibinfo{journal}{\emph{arXiv preprint arXiv:2302.04473}} (\bibinfo{year}{2023}).
\newblock


\bibitem[Zhou(2023)]%
        {zhou2023mmrecsm}
\bibfield{author}{\bibinfo{person}{Xin Zhou}.} \bibinfo{year}{2023}\natexlab{}.
\newblock \showarticletitle{MMRec: Simplifying Multimodal Recommendation}.
\newblock \bibinfo{journal}{\emph{arXiv preprint arXiv:2302.03497}} (\bibinfo{year}{2023}).
\newblock


\bibitem[Zhou et~al\mbox{.}(2023a)]%
        {zhou2023layer}
\bibfield{author}{\bibinfo{person}{Xin Zhou}, \bibinfo{person}{Donghui Lin}, \bibinfo{person}{Yong Liu}, {and} \bibinfo{person}{Chunyan Miao}.} \bibinfo{year}{2023}\natexlab{a}.
\newblock \showarticletitle{Layer-refined graph convolutional networks for recommendation}. In \bibinfo{booktitle}{\emph{2023 IEEE 39th International Conference on Data Engineering (ICDE)}}. IEEE, \bibinfo{pages}{1247--1259}.
\newblock


\bibitem[Zhou and Shen(2023)]%
        {zhou2023tale}
\bibfield{author}{\bibinfo{person}{Xin Zhou} {and} \bibinfo{person}{Zhiqi Shen}.} \bibinfo{year}{2023}\natexlab{}.
\newblock \showarticletitle{A tale of two graphs: Freezing and denoising graph structures for multimodal recommendation}. In \bibinfo{booktitle}{\emph{Proceedings of the 31st ACM International Conference on Multimedia}}. \bibinfo{pages}{935--943}.
\newblock


\bibitem[Zhou et~al\mbox{.}(2023b)]%
        {zhou2023bootstrap}
\bibfield{author}{\bibinfo{person}{Xin Zhou}, \bibinfo{person}{Hongyu Zhou}, \bibinfo{person}{Yong Liu}, \bibinfo{person}{Zhiwei Zeng}, \bibinfo{person}{Chunyan Miao}, \bibinfo{person}{Pengwei Wang}, \bibinfo{person}{Yuan You}, {and} \bibinfo{person}{Feijun Jiang}.} \bibinfo{year}{2023}\natexlab{b}.
\newblock \showarticletitle{Bootstrap latent representations for multi-modal recommendation}. In \bibinfo{booktitle}{\emph{Proceedings of the ACM Web Conference 2023}}. \bibinfo{pages}{845--854}.
\newblock


\end{thebibliography}


\newpage

\appendix
 
\section{Proof of Theorem~\ref{thm:ncer} (NCER Noise Reduction)}
\label{app:proof_ncer}
 
\begin{proof}
Decompose the observed graph as $\bm{S} = \bm{S}^* + \bm{N}$. The NCER reweighting applies element-wise multiplication with $\bm{W}$ where $W_{ij} = 1 + \kappa \cdot c_{ij}$:
\begin{align}
\hat{\bm{S}} &= \bm{S} \circ \bm{W} = (\bm{S}^* + \bm{N}) \circ \bm{W} \\
&= \bm{S}^* \circ \bm{W} + \bm{N} \circ \bm{W} = \hat{\bm{S}}^* + \hat{\bm{N}}.
\end{align}
 
Consider the Frobenius norms. For any matrix $\bm{M}$ with element-wise reweighting:
\begin{equation}
\|\bm{M} \circ \bm{W}\|_F^2 = \sum_{(i,j): M_{ij} \neq 0} M_{ij}^2 (1 + \kappa c_{ij})^2.
\end{equation}
 
For the noise component, define the average consistency on noise edges:
\begin{equation}
\bar{c}_N = \frac{\sum_{(i,j) \in \text{supp}(\bm{N})} N_{ij}^2 c_{ij}}{\sum_{(i,j) \in \text{supp}(\bm{N})} N_{ij}^2}.
\end{equation}
 
By the Cauchy-Schwarz inequality applied to the weighted sum:
\begin{equation}
\|\hat{\bm{N}}\|_F^2 \leq \|\bm{N}\|_F^2 \cdot (1 + \kappa \bar{c}_N)^2 + O(\kappa^2 \text{Var}(c_N)),
\end{equation}
where the second-order term accounts for variance in consistency scores. Similarly for the signal:
\begin{equation}
\|\hat{\bm{S}}^*\|_F^2 \geq \|\bm{S}^*\|_F^2 \cdot (1 + \kappa \bar{c}_{S^*})^2.
\end{equation}
 
Therefore, the noise-to-signal ratio satisfies:
\begin{equation}
\frac{\|\hat{\bm{N}}\|_F}{\|\hat{\bm{S}}^*\|_F} \leq \frac{\|\bm{N}\|_F}{\|\bm{S}^*\|_F} \cdot \frac{1 + \kappa \bar{c}_N}{1 + \kappa \bar{c}_{S^*}} + O(\kappa^2).
\end{equation}
 
Under the community structure assumption (Assumption~\ref{ass:noisy_graph}), clean edges connect items within the same semantic community and thus share many neighbors ($\bar{c}_{S^*}$ is large), while noise edges connect items across communities with few shared neighbors ($\bar{c}_N$ is small). Therefore $\bar{c}_{S^*} > \bar{c}_N$, and the multiplicative factor $(1 + \kappa \bar{c}_N)/(1 + \kappa \bar{c}_{S^*}) < 1$ for all $\kappa > 0$.
 
Extending from $\bm{S}^*$ to $\bm{S}$ in the denominator introduces an additional multiplicative correction bounded by $(1 + \epsilon/\|\bm{S}^*\|_F)^{-1}$, yielding the stated bound.
\end{proof}

\section{Proof of Theorem~\ref{thm:rig} (RIG Non-Degeneracy)}
\label{app:proof_rig}
 
\begin{proof}
The BPR loss for a triplet $(u, p, n)$ is:
\begin{equation}
\cL_{u,p,n} = -\log \sigma(\be_u^\top \be_p - \be_u^\top \be_n).
\end{equation}
 
With RIG, $\be_i = \be^{\text{UI}}_i + g_i \cdot \bh_i$. Taking the derivative with respect to $g_p$ (gate of positive item):
\begin{align}
\frac{\partial \cL_{u,p,n}}{\partial g_p} &= -\sigma(-\Delta_{u,p,n}) \cdot \be_u^\top \frac{\partial \be_p}{\partial g_p} \\
&= -\sigma(-\Delta_{u,p,n}) \cdot \be_u^\top \bh_p,
\end{align}
where $\Delta_{u,p,n} = \be_u^\top \be_p - \be_u^\top \be_n$ and $\sigma(-\Delta) = 1 - \sigma(\Delta) \in (0,1)$.
 
Define $w_{u,p} = \sigma(-\Delta_{u,p,n}) > 0$. Aggregating over all triplets where $i$ appears as positive:
\begin{equation}
\frac{\partial \cL}{\partial g_i}\bigg|_{\text{positive}} = -\sum_{u: (u,i,\cdot) \in \cD} w_{u,i} \cdot \be_u^\top \bh_i.
\end{equation}
 
Similarly, when $i$ appears as a \emph{negative} item in triplet $(u, p, i)$:
\begin{align}
\frac{\partial \cL_{u,p,i}}{\partial g_i} &= -\sigma(-\Delta_{u,p,i}) \cdot \be_u^\top \left(-\frac{\partial \be_i}{\partial g_i}\right) \\
&= \sigma(-\Delta_{u,p,i}) \cdot \be_u^\top \bh_i.
\end{align}
 
Define $w'_{u,i} = \sigma(-\Delta_{u,p,i}) > 0$. Aggregating:
\begin{equation}
\frac{\partial \cL}{\partial g_i}\bigg|_{\text{negative}} = \sum_{u: (u,\cdot,i) \in \cD} w'_{u,i} \cdot \be_u^\top \bh_i.
\end{equation}
 
At a stationary point, $\partial \cL / \partial g_i = 0$, which requires:
\begin{equation}
\sum_{u \in \text{pos}} w_{u,i} \cdot \be_u^\top \bh_i = \sum_{u \in \text{neg}} w'_{u,i} \cdot \be_u^\top \bh_i.
\end{equation}
 
\textbf{Case 1: $\be_u^\top \bh_i > 0$ on average (helpful II signal).} The positive-side gradient pushes $g_i$ upward (increasing $g_i$ reduces $\cL$), while the negative-side gradient pushes $g_i$ upward as well (increasing $g_i$ increases the negative item's score, which increases $\cL$). The balance is achieved at an interior point because the BPR weights $w, w'$ depend on the margin: as $g_i$ increases, the margins change, adjusting the weights until equilibrium.
 
\textbf{Case 2: $\be_u^\top \bh_i < 0$ on average (harmful II signal).} The roles reverse, and the equilibrium $g_i$ is pushed toward lower values.
 
Since the sigmoid output of the gate MLP is continuous and strictly bounded in $(0,1)$ (assuming bounded inputs and non-degenerate initialization), gradient descent converges to $g_i^* \in (0,1)$ that balances the positive and negative contributions. Furthermore, if there exist items $i, j$ with $\mathbb{E}[\be_u^\top \bh_i] > 0 > \mathbb{E}[\be_u^\top \bh_j]$, the equilibrium conditions require $g_i^* > g_j^*$.
\end{proof}

\section{Proof of Theorem~\ref{thm:ina} (INA Generalization Bound)}
\label{app:proof_ina}
 
\begin{proof}
The standard PAC-Bayes bound for BPR with $n$ observed interactions states:
\begin{equation}
R(f) \leq \hat{R}_n(f) + \sqrt{\frac{\text{KL}(f \| \pi) + \log(2\sqrt{n}/\delta)}{2n}},
\end{equation}
with probability at least $1-\delta$ over the draw of the training set.
 
INA augments the training set with $n_{\text{aug}} = n \cdot 1$ additional soft-positive triplets (one per observed interaction). Each augmented triplet $(u, j, n')$ satisfies $j \in \cN_{\text{II}}(p)$ where $(u, p)$ is a true interaction.
 
The key question is the \emph{effective sample size} contributed by augmented samples. Define the \emph{soft label quality}:
\begin{equation}
q = \mathbb{P}[\text{user} u \text{prefers} j \text{ over random } n' | j \in \cN_{\text{II}}(p), (u,p) \in \bm{R}^+].
\end{equation}
 
This is the probability that a neighbor of a positive item is itself a relevant item. The discount factor $\delta$ in INA accounts for this uncertainty. Following the analysis of learning with noisy labels~\cite{natarajan2013learning}, each augmented sample contributes an effective fraction $\delta_{\text{eff}} = \delta^2 \cdot q$ to the sample size. This arises because the discount $\delta$ reduces both the gradient magnitude (linear in $\delta$) and the effective signal strength (linear in $\delta \cdot q$).
 
Substituting the effective sample size $n + \delta_{\text{eff}} \cdot n_{\text{aug}}$ into the PAC-Bayes bound:
\begin{equation}
R(f) \leq \hat{R}_{n+n_{\text{aug}}}(f) + \sqrt{\frac{\text{KL}(f \| \pi) + \log(2\sqrt{n + \delta_{\text{eff}} n_{\text{aug}}}/\delta)}{2(n + \delta_{\text{eff}} n_{\text{aug}})}}.
\end{equation}
 
When the II graph has positive predictive value ($q > 0$) and the discount is positive ($\delta > 0$), we have $\delta_{\text{eff}} > 0$ and thus $n + \delta_{\text{eff}} n_{\text{aug}} > n$, yielding a strictly tighter bound than standard BPR.
\end{proof}

\section{Baselines}
\label{appendix:baselines}
We compare IIMRec against a diverse set of representative baselines, categorized into two groups.

\noindent1) General recommendation models:
\begin{itemize}[leftmargin=*]
\item \textbf{LightGCN} \cite{he2020lightgcn} simplifies graph convolutional network (GCN)-based collaborative filtering by eliminating redundant components such as feature transformation and nonlinear activation.
\item \textbf{SimGCL} \cite{yu2022graph} performs graph contrastive learning by injecting random noise directly into the representation space instead of augmenting the graph structure.
\item \textbf{LayerGCN} \cite{zhou2023layer} introduces residual connections into GCN layers to mitigate the over-smoothing problem in LightGCN.
\item \textbf{NLGCL} \cite{xu2025nlgcl} leverages naturally occurring high-quality contrastive signals on graphs to enhance recommendation accuracy.
\end{itemize}
\noindent2) Multimodal recommendation models:
\begin{itemize}[leftmargin=*]
\item \textbf{MMGCN} \citep{wei2019mmgcn} employs separate GCNs to process and integrate information from different modalities.
\item \textbf{DualGNN} \citep{wang2021dualgnn} constructs an auxiliary user-user graph to capture user preference correlations alongside multimodal item features.
\item \textbf{LATTICE} \citep{zhang2021mining} builds item-item semantic graphs from multimodal features to capture latent correlative signals among items.
\item \textbf{SLMRec} \citep{tao2022self} applies self-supervised learning to multimodal recommendation via feature-level noise perturbation and cross-modal pattern discovery.
\item \textbf{BM3} \citep{zhou2023bootstrap} simplifies self-supervised multimodal recommendation by generating contrastive views through dropout-based perturbation.
\item \textbf{MMSSL} \citep{wei2023multi} employs modality-aware adversarial perturbations for interactive structure learning and combines cross-modal contrastive learning to disentangle shared and modality-specific features.
\item \textbf{FREEDOM} \citep{zhou2023tale} denoises the user-item interaction graph and constructs a frozen item-item graph from raw modality features.
\item \textbf{LGMRec} \citep{guo2024lgmrec} utilizes hypergraph structures to capture both local topological information and global collaborative embeddings.
\item \textbf{DiffMM} \citep{jiang2024diffmm} integrates a modality-aware graph diffusion model with cross-modal contrastive learning to improve user representation learning.
\item \textbf{SMORE} \citep{ong2025spectrum} mitigates modality noise by leveraging discriminative spectral properties in the frequency domain.
\item \textbf{HPMRec} \citep{chen2025hypercomplex} enriches feature diversity and bridges semantic gaps across modalities through hypercomplex representations.
\item \textbf{MENTOR} \citep{xu2024mentor} designs multi-level cross-modal alignment tasks to enhance the quality of learned representations.
\item \textbf{COHESION} \citep{xu2025cohesion} introduces a dual-stage fusion mechanism to fully exploit composite graph structures.
\end{itemize}

\section{Algorithm}
\label{appendix:alg}
The complete training procedure is summarized in Algorithm~\ref{alg:iimrec}.
\begin{algorithm}[t]
\caption{IIMRec Training Procedure}
\label{alg:iimrec}
\begin{algorithmic}[1]
\REQUIRE Interaction matrix $\bm{R}$, multimodal features $\{\bm{F}^m\}$
\STATE \textbf{Preprocessing (one-time):}
\STATE Construct semantic graph $\bm{S}_{\text{sem}}$ (spatial + spectral KNN)
\STATE Construct co-occurrence graph $\bm{S}_{\text{co}}$ from $\bm{R}^\top \bm{R}$
\STATE Fuse: $\bm{S}_{\text{fused}} = \bm{S}_{\text{co}} + \bm{S}_{\text{sem}}$
\STATE Apply NCER reweighting (Eq.~\eqref{eq:ncer_reweight}) $\to \bm{S}_{\text{NCER}}$
\STATE Build content-guided UI graph $\tilde{\bm{A}}$
\STATE Pre-compute INA neighbor table from $\bm{S}_{\text{NCER}}$
\FOR{epoch $= 1$ to max\_epochs}
    \FOR{each mini-batch $(u, p, n) \in \cD$}
        \STATE Project modality features (Eq.~\eqref{eq:projection})
        \STATE Propagate on UI graph (Eq.~\eqref{eq:residual_prop})
        \STATE Fuse user modality representations (Eq.~\eqref{eq:user_fusion})
        \STATE Propagate item representations on II graph
        \STATE Apply RIG gating (Eq.~\eqref{eq:rig})
        \STATE Calculate contrastive learning losses (Eq.~\eqref{eq:cl_user})
        \STATE Compute $\cL_{\text{BPR}} + \lambda_{\text{INA}} \cL_{\text{INA}} + \lambda_{\text{CL}} \cL_{\text{CL}}$ 
        \STATE Update parameters via backpropagation
    \ENDFOR
\ENDFOR
\end{algorithmic}
\end{algorithm}

\begin{figure}[!t] 
    \centering
    \includegraphics[width=1\linewidth]{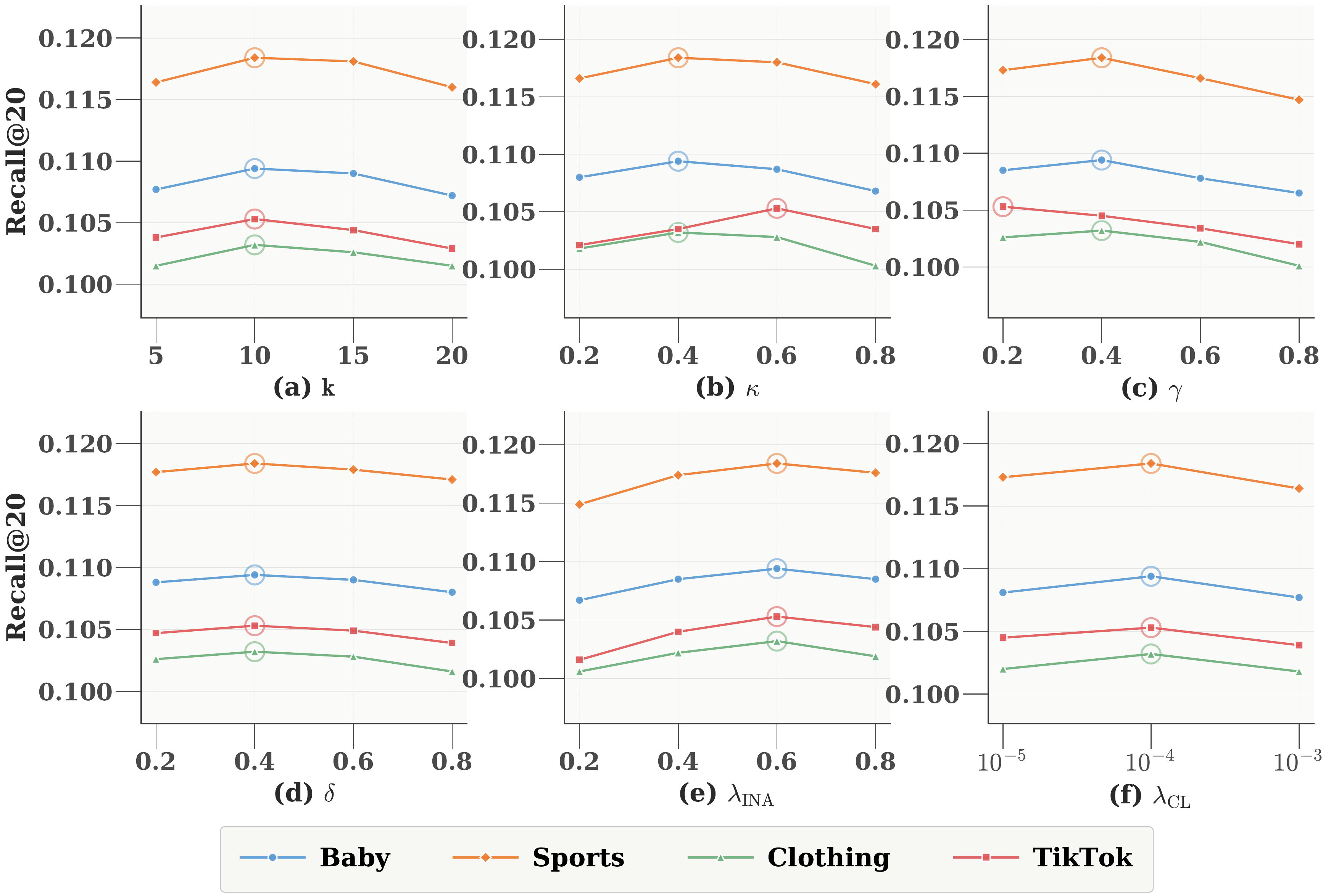}
    \vskip -0.15in
    \caption{Hyper-parameter analysis for IIMRec across all datasets in terms of Recall@20.}
    \label{fig:hyperparameter}
    \vskip -0.15in
\end{figure}

\subsection{Hyper-parameter Analysis}
\label{appendix:hyper}
We investigate the sensitivity of IIMRec to six key hyper-parameters by varying each while fixing the others at their optimal values, and report Recall@20 across all four datasets in Figure~\ref{fig:hyperparameter}.
 
\noindent \textbf{KNN neighbor count $k$.} Performance peaks at $k{=}10$ across all datasets and degrades at both extremes: too few neighbors limit the semantic coverage of the item-item graph, while too many introduce noisy edges that dilute the propagated signal.
 
\noindent \textbf{NCER reweighting strength $\kappa$.} A moderate $\kappa{=}0.4$ yields the best results on three datasets, with TikTok favoring $\kappa{=}0.6$. Larger values over-amplify high-overlap edges at the expense of useful but structurally weaker connections.
 
\noindent \textbf{Content-guided UI balance $\gamma$.} All datasets favor $\gamma \in \{0.2, 0.4\}$, indicating that a modest injection of virtual interactions benefits collaborative propagation, but excessive reliance on the content-guided graph introduces noise.
 
\noindent \textbf{INA discount factor $\delta$.} Performance is relatively stable across $\delta \in [0.2, 0.6]$ and declines at $\delta{=}0.8$, confirming that soft positives should be discounted to avoid over-weighting uncertain supervision.
 
\noindent \textbf{INA loss weight $\lambda_{\text{INA}}$.} The optimal value is $\lambda_{\text{INA}}{=}0.6$ across all datasets. Too small a weight under-utilizes the augmented signal, while too large a weight lets soft-positive gradients dominate the primary BPR objective.
 
\noindent \textbf{Contrastive loss weight $\lambda_{\text{CL}}$.} $\lambda_{\text{CL}}{=}10^{-4}$ is consistently optimal. A larger weight ($10^{-3}$) forces excessive cross-modal alignment that conflicts with the collaborative filtering signal.

\end{document}